\documentclass[10pt,twocolumn,graphicx,psfig,amssymb,amsthm,amsmath]{IEEEtran}
\IEEEoverridecommandlockouts
\usepackage{cite}
\usepackage{amsmath,amssymb,amsfonts}
\usepackage{amsthm}
\usepackage{bm}
\usepackage{makecell}
\usepackage{enumitem}
\makeatletter
\def\@upn{}
\makeatletter
\usepackage{algorithmic}
\usepackage{graphicx}
\usepackage{bbm}
\usepackage{textcomp}
\usepackage{xcolor}
\theoremstyle{plain}
\newtheoremstyle{italicheader}
{3pt}                
{3pt}                
{ \normalfont}       
{}                   
{\itshape}           
{.}                  
{.05em}               
{}                   
\newtheorem{theorem}{Theorem}
\newtheorem{remark}{Remark}[theorem]
\newtheorem{corollary}{Corollary}[theorem]
\usepackage[left=0.56in,right=0.56in]{geometry}
\newgeometry{textwidth=7.25in,top=0.7in, bottom=1.05in}
\def\BibTeX{{\rm B\kern-.05em{\sc i\kern-.025em b}\kern-.08em
		T\kern-.1667em\lower.7ex\hbox{E}\kern-.125emX}}
	\DeclareMathOperator\erfc{erfc}
	\DeclareMathOperator\erf{erf}
\begin{document}
	\title{Multi-Axis Concentration Modulation for Mobile Molecular Communication Systems
	}
	\author {Muskan Ahuja 
		 and Abhishek K. Gupta 
		 \vspace*{-3em}
		\thanks{This research is supported by Qualcomm (6G University Research), IITK and SERB India (Grant 2022/0390 and CRG/2023/005206).
			
			Muskan Ahuja is with the Computer Science and Engineering Dept., South Asian University (e-mail: muskanahuja@sau.int) and Abhishek K. Gupta is with the Electrical Engineering Dept., IIT Kanpur, India (e-mail: gkrabhi@iitk.ac.in). A part of this work has been presented in IEEE WCNC 2025 \cite{minee}.}}
	\maketitle
	\begin{abstract}
Molecular communication (MC) is emerging
paradigm that employs molecules as information carriers, inspired
by biological signaling processes. Existing modulation schemes
such as on–off keying (OOK), although simple to implement, suffer from high
error probability in dynamic or hard-to-estimate channels due to
their dependence on accurate channel information (CI).
This work develops a unified MC constellation framework that allows higher order modulation across multiple dimensions and designs efficient constellation for dynamic MC.
We propose a general multi-axis concentration modulation (MAxCM($K$,$M$)) of modulation order $M$, utilizing $K$-dimensional constellation space with each axis corresponding to a particular molecular type, and information is jointly encoded in their concentrations. The corresponding ML decoders are derived for both static and dynamic MC under exact and partial CI. We show that the use of MAxCM can provide improvements in spectral efficiency (SE) and error rate. We then focus on a special subclass, namely multiple-axis ratio shift keying (MAxRSK), that encodes information into the concentration ratios. Its ML decoder is shown to be a weighted combiner, and design constraints are derived to enable channel-independent decoding. We study one such example, symmetric binary RSK (SBRSK), to show its robustness in dynamic channel conditions compared to OOK.
Numerical investigations show significant performance gains over OOK and provide insights into optimal constellation design and receiver configurations. 
	\end{abstract}
	\section{Introduction}
	 MC is a form of bio-inspired communication that uses molecules as information carriers, mimicking how biological systems, such as cells, communicate in nature. MC involves the transmission of information via signaling molecules (SMs) from a transmitter (TX) to a receiver (RX) via a medium utilizing a propagation mechanism such as diffusion \cite{survey}. Owing to its biocompatibility and energy efficiency, MC shows potential for connecting different bio-nanodevices, such as engineered bacteria and nano biosensors. This enables new healthcare applications, including targeted drug delivery, disease detection, tissue engineering, and continuous health monitoring within the internet of bio-nano things (IoBNT) network \cite{iobnt}.
	
		While traditional MC models often assume TX and RX nanomachines to be static, many realistic scenarios involve mobile molecular communication (MMC), where TXs and RXs can move through the medium over time. This dynamic setting aligns more closely with biological environments, where nanoscale devices such as mobile bacteria, vesicles, nanorobots, or
		drug carriers, naturally move within fluidic medium. The movement of devices can improve information transfer rates under certain conditions by improving signal reception via increased propagation \cite{motors}, \cite{mul}.
		The MMC paradigm enables a broad array of novel applications, including network functionalities (e.g., relay or multi-hop communication), biomolecular sensing in dynamic environments, and critical biomedical applications, such as cooperative drug delivery, where multiple mobile nanomachines coordinate to target specific regions more effectively \cite{randomwalk}.
		However, MMC brings many challenges that are absent in static systems.
		These include difficulty in deriving precise channel models due to constant variation in distance between devices and inherent randomness in their motions. Moreover, the direction and distance between MMC nodes are constantly changing. The randomness in the locations of TXs and RXs can deteriorate the reliability of MMC. Understanding the impact of mobility on the performance of MC is crucial to design effective MMC systems, including efficient modulation methods, and decoding strategies \cite{challengeMMC}. 
		\vspace{-1em}
	\subsection{Related work}	
	In the past, researchers have worked on the MC channel characterization \cite{chn}, performance, and information-theoretic analyses of different modulation schemes with different types of MC RXs \cite{markov, markov2, abhisheksir}.  
	Various modulation techniques may be used in MC to encode data on different properties of SMs, such as their concentration \cite{c2}, time of release \cite{time}, \cite{t2}, their type \cite{type}, \cite{t1}, index of the TX \cite{spatio0, spatio1, spatio2} or any combination of these \cite{h1}, \cite{mod}, \cite{c1}. 
	 In particular, concentration based modulation schemes encode the information using concentration of the same or different types of molecules. Out of these, on-off keying (OOK) modulation where a certain number of molecules are sent for bit 1 compared to no transmission for bit 0, has gained popularity due to its simplicity and efficiency. Systems with OOK modulation utilize a threshold-based decoder which requires the exact knowledge of channel to determine the threshold. In the static scenario with a known environment, where TXs and RXs are stationary, the CI can be obtained over time. However, in realistic scenarios, TX and/or RX may not be stationary and can move in the medium via some mechanism including Brownian motion \cite{randomwalk}.
	The mobility leads to a variation in the distance between TX and RX units and hence, in the channel impulse response (CIR) with time \cite{stochastic}. Hence, an exact knowledge of the channel may
	not be available or maybe difficult to evaluate, adding another layer of complexity.
	   Therefore, the OOK-based MC systems suffer from high bit error probability in unknown/ difficult to estimate channel environments. Such dynamic scenarios demand new modeling, energy efficient modulation, and detection techniques that can accommodate and even exploit the variability in the communication channel caused by mobility.
	   
	   The work in \cite{novelRSK} proposed ratio-shift keying (RSK) modulation that encodes information in the ratio of the concentration of two isomers of SMs and the receiver decodes the symbol utilizing the ratio of their received concentration. Due to almost identical shapes, isomers demonstrate similar diffusion coefficients in a medium and, therefore, the same propagation characteristics. The authors argued that the relative concentrations of two molecules' type remains constant regardless of the channel value, therefore, RSK does not require CI for decoding. Due to its inherent robustness to channel uncertainties arising due to the various reasons including movement of transceivers, RSK can be utilized for dynamic MC scenarios. However, a detailed analysis
	   on the performance of RSK in dynamic scenario was not presented. Authors in \cite{rskmain} studied the bit error rate (BER) performance of RSK for an MC system with a ligand-receptor RX that measures its receptor’s bound time duration to estimate the transmit symbol’s value. While the paper considers a simple threshold/count detector as decoder for OOK modulation, it utilizes the exact bound/unbound time durations to decode for RSK modulation. Due to additional complexities required in measuring the bound state durations, this comparison as performed in \cite{rskmain} may be considered unfair, and may bring additional difficulties in realizing such systems. It may be more interesting to investigate the performance of RSK modulation in dynamic scenarios under a similar decoding mechanism as OOK. 
	   An extension to multiple types of molecules was studied in \cite{mrsk} for static MC systems. While the study represents a step toward utilizing ratio-based signaling to reduce the need for CI, the study was limited to static MC system and did not analyze the performance for mobile scenarios. Another limitation of \cite{mrsk} is the lack of normalization in total number of molecules utilized for each symbol in comparison with OOK. As a consequence, the RSK in \cite{mrsk} inherently utilized a larger number of molecules per transmitted symbol than simpler schemes like OOK, leading to an unfair comparison where RSK appeared to outperform OOK with large margin. It is important to study the performance under normalization constraints over total molecular budget and bit time duration, to offer a fair and resource-aware comparison across different modulation schemes.
	   Further, above works have considered a noise-less environment. In practical scenario, there can be unintended molecules present in the medium due to previous transmissions or other sources. In the presence of noise, the ratio of the concentration of two types of molecules no longer remains invariant of the channel. Therefore, for an efficient use of RSK in dynamic MC, it remains crucial to explore how to design constellation to maintain channel invariance in the presence of noise which was not discussed in the past literature. 

Another important aspect of designing an efficient constellation for a modulation is to accommodate as many symbols as possible. While more symbols can be placed within the same constellation space at the expense of reduced separation between symbols, resulting in the loss of error performance, the use of additional orthogonal signals can expand the constellation space. In conventional wireless systems, orthogonal basis signals are often realized through differences in phase or frequency which form a signal space. The use of multiple basis vectors provides opportunity to accommodate more symbols in the constellation space resulting in a higher modulation order. The decoding of these symbols typically requires accurate knowledge of channel, as channel
impairments can rotate or distort the constellation, leading to symbol misclassification and increased decoding errors \cite{dc}. Example of modulations utilizing multiple orthogonal basis include quadrature amplitude modulation (QAM) and orthogonal frequency divison modulation (OFDM). 
These aspects motivated us to exploit multiple orthogonal signal dimensions in MC, where orthogonality can be realized through different types of SMs that propagate with similar physical characteristics but are distinguishable at the receiver. The use of similar SMs provides us an opportunity to utilize their ratio to mitigate the effect of channel variations on the symbols' reception. Furthermore, unlike phase or frequency-based schemes in wireless systems, where a signal may rotate due to channel effects and
decoding needs accurate CI, the molecular types in MC can be taken to be immutable (e.g., molecules of a type usually do not transform into another type during propagation), making the scheme inherently resistant to rotational errors in the molecular signal space. Therefore, it may be possible to decode the symbol without requiring CI. This indicates the existence of a multidimensional signal space which can be efficiently harnessed for designing optimal signal constellation for MC systems in the presence of mobility or channel uncertainity. Therefore, the design of appropriate signal constellations tailored to the molecular medium and its multidimensional nature is crucial for improving the reliability, efficiency, and scalability of MMC systems, which is the main focus of this work.
	\vspace{-1.2em}
	\subsection{Contributions}
	While keeping a two-fold aim to develop a generalized constellation framework allowing the use of multiple dimensions in the constellation space with a higher order of modulation and to design an efficient constellation for a dynamic MC environment, in this work, we propose multi-axis concentration modulation (MAxCM) and present the corresponding ML decoder. We also design a special case of the proposed modulation where the corresponding decoder is invariant to the channel conditions, even in the presence of noise, making it suitable for a dynamic MC environment. The key contributions of this paper are as follows:
	\begin{itemize}[leftmargin=0.3cm]
		\item In this work, we propose MAxCM utilizing multi-axis constellation space where each orthogonal axis corresponds to a different type of SM. Here, the information is encoded jointly into the concentration/molecules-count of different types of molecules. This can be equivalently seen as encoding information into the sum concentration and ratio of concentrations of various types, analogous to amplitude and phase in QAM modulation in electromagnetic (EM) based communication. The proposed constellation is flexible and general enough to support arbitrary modulation order and dimension of constellation space. The framework naturally generalizes well-known schemes such as OOK, concentration shift keying (CSK) and RSK as special cases and provides a unified approach to study these modulations. In addition, it provides many novel modulations which has not been studied in prior MC work.
	\item We also present the corresponding ML decoder for static and dynamic MC systems for cases when either exact or partial information about the channel is available. Moreover, we show that the use of MAxCM can provide improvements in spectral efficiency (SE) and error rate.
	\item We then focus on a special subclass of MAxCM, namely multiple-axis ratio shift keying (MAxRSK), when information is encoded in the ratio of types of SMs. We first derive the corresponding ML decoder and show that it is a weighted combiner with weights dependent on the strength of the received signals. We characterize the decision regions for each symbol and show that the decision boundaries form a line/plane passing through the origin.
	\item We then derive constraints on the constellation symbols that allow the decoder not to require any CI. We present a few examples of such constellations, including symmetric binary RSK (SBRSK), and derive their bit error probability performance for static and dynamic environments considering both passive and fully absorbing receivers. We show the robustness of SBRSK in dynamic channel conditions compared to OOK while maintaining a comparable decoding complexity under constraints on molecular budget and bit duration.
	\item We demonstrate an approach for an extension to higher-order symmetric MAxRSK (SMAxRSK) while preserving channel-independent decoding properties. Importantly, the constellation is designed using geometric symmetry and 
	we provide rules to construct such constellation along with construction of specific examples.
	 This highlights the expressive power and scalability of the proposed MAxRSK framework beyond previously studied binary schemes.
	\item We also discuss practical cases when exact CI is not known. We show that the performance of SBRSK remains unaffected in such cases, whereas the performance of OOK significantly decreases when only channel statistical information (CSI) is available. We also evaluate the effect of imperfect CSI on OOK.
	\item We also show the variation of bit error probability of SBRSK with the number of transmitted molecules (molecular resource) for both passive and absorbing receivers in MMC systems. Extensive numerical investigations are conducted to evaluate system performance and extract design insights, such as optimal symbol construction, the performance gain of SBRSK over OOK as a function of node mobility, and the impact of receiver type and signal strength. Together, these features make SMAxRSK a compelling modulation strategy for future MC applications, especially in mobile and biologically dynamic environments.
	\end{itemize}
	\vspace{-1em}
	\section{System Model}
	In this paper, we consider a MC system with a TX and a RX located in a three-dimensional environment where a TX aims to convey information to the RX using SMs. Let the positions of TX and RX at time $t$ be denoted by ($X_{\text{TX}}(t), Y_{\text{TX}}(t), Z_{\text{TX}}(t)$) and ($X_{\text{RX}}(t), Y_{\text{RX}}(t), Z_{\text{RX}}(t)$) respectively.
	The RX decodes information based on the molecules it observes over time. The probability that a molecule released by the TX at time $t$ is observed at the RX after a time interval $\tau_{s}$ is denoted by $h(t, \tau_{s})$. This probability characterizes the channel response and depends on factors such as molecular dynamics, medium characteristics, and receiver type. For example, for an absorbing receiver, each molecule that reaches the RX boundary is absorbed and counted, after which it no longer exists in the environment. Conversely, a passive receiver continuously measures the concentration of molecules within its observation volume without absorption. To allow the use of multiple types of SMs, we consider the TX to be capable of releasing $K$ types of SMs with types $T_{1}$ to $T_{K}$. Similarly, the RX unit comprises of $K$ RX ports (RXPs), embedded with receptors specific to different molecular types.
	We consider SMs of all types with identical propagation characteristics so the effect of channel is the same for all molecules. One practical realization can be achieved using $K$ molecular isomers that differ structurally but have identical propagation characteristics while enabling type-specific detection at the RX. However, note that the proposed framework can be extended to accommodate molecules with distinct propagation characteristics. Let $r_{\text{RX}}$ denote the radius of each RX port. We consider both static and dynamic communication scenarios. In the static scenario, both TX and RX remain fixed in the environment. In the dynamic MC, they both can move within the medium. Hence, locations of TX and RX are not known a priori and they need to be estimated to determine the channel. The received signal is affected by noise due to molecules from the unintended sources present in the environment. Each molecular type is subjected to an independent noise and can be modeled as a Poisson random variable. 
		\begin{figure}[h]
		\centering
		\includegraphics[width=0.75\linewidth]{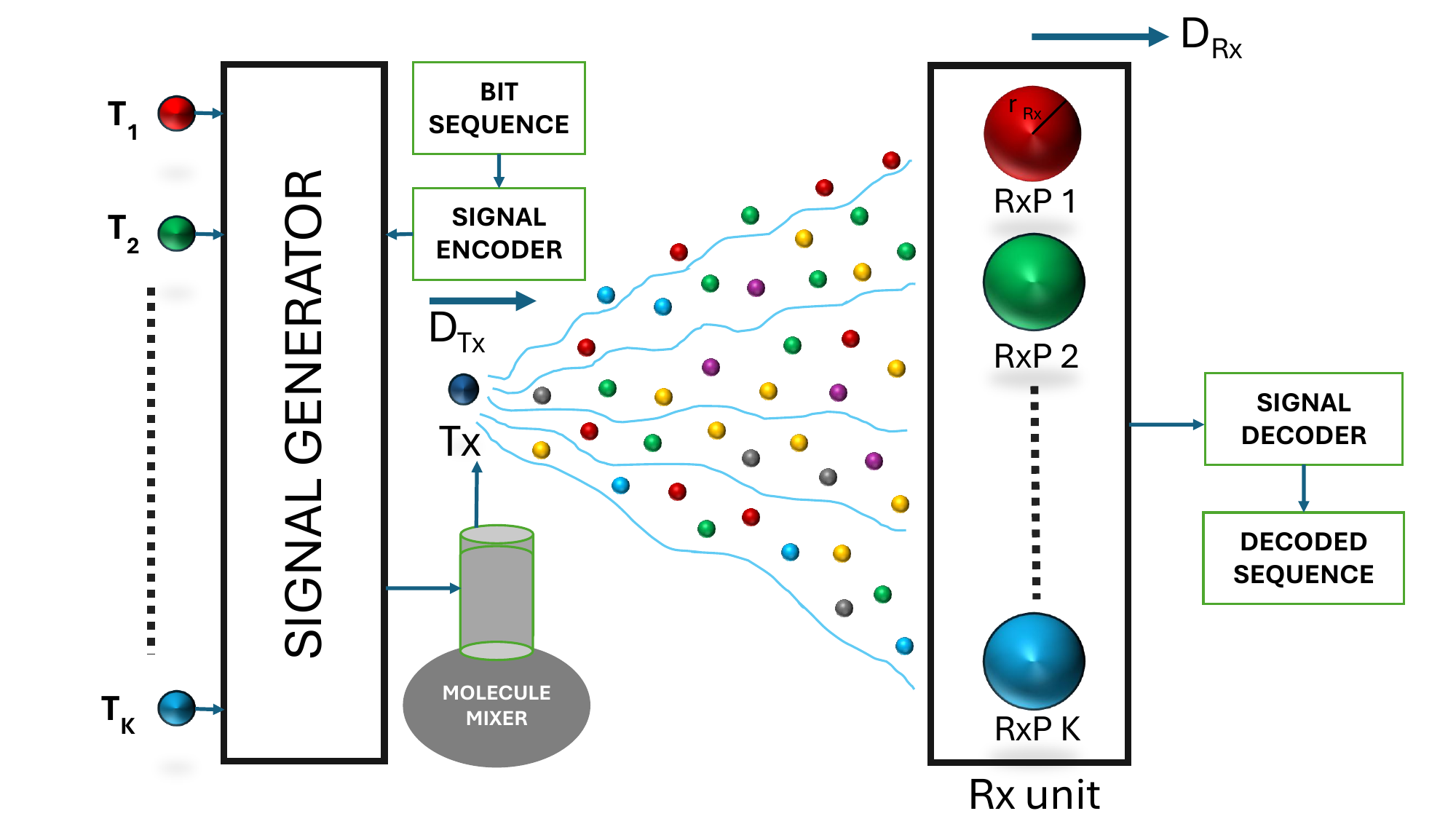}
		\caption{Demonstration of the considered system model for the MMC scenario employing MAxCM modulation scheme. The point TX is shown as a small blue sphere. The figure depicts the red sphere for Type 1 RXP, the green sphere for Type 2 RXP, and the light blue sphere for Type $K$ RXP. $D_{\text{TX}}$ and $D_{\text{RX}}$ are the diffusion coefficients of TX and RX units respectively.}
		\label{figure1}
	\end{figure}
	
	To simplify the system representation and analyze further, we will consider a diffusion-based MC scenario. In this scenario, the molecules propagate via Brownian motion, characterized by a diffusion coefficient denoted as $D_{\text{inf}}$. Note that the diffusion coefficient can vary depending on several factors like concentration gradient, temperature, medium properties, and other external factors \cite{survey}. We have considered diffusion coefficient to remain constant. Further, for the dynamic case, where TX and RX can move, we consider that they move via Brownian motion, which is a common approach for modeling the movement of different microorganisms and cells \cite{randomwalk}. In particular, we consider a transparent mobile point TX moving in the medium via diffusion with diffusion coefficient $D_{\text{TX}}$.
	The RX's mobility is denoted by diffusion coefficient $D_{\text{RX}}$ (cf. Fig. \ref{figure1}). For simplicity, we assume $D_{\text{RX}}= D_{\text{TX}}\stackrel{\Delta}{=} D_\mathrm{TR}$.
	At time $t$, each coordinate of TX ($X_{\text{TX}}(t), Y_{\text{TX}}(t), Z_{\text{TX}}(t)$) and RX ($X_{\text{RX}}(t), Y_{\text{RX}}(t), Z_{\text{RX}}(t)$) is distributed as gaussian random variable with variance $2D_\mathrm{TR}t$ around the initial position \cite{rskmain}.
	Let ($x^{0}_{\text{TX}}, y^{0}_{\text{TX}}, z^{0}_{\text{TX}}$) and ($x^{0}_{\text{RX}}, y^{0}_{\text{RX}}, z^{0}_{\text{RX}}$) denote the initial position of TX and RX unit, and $\mathcal{N}(\mu,\sigma^{2})$ denotes the normal distribution with mean $\mu$ and variance $\sigma^{2}$. Moreover, we assume that initially the TX is located at the origin and the RX unit is located at ($x^{0}_{\text{RX}}=r(0), y^{0}_{\text{RX}}=0, z^{0}_{\text{RX}}=0$).

	 The probability distribution function (PDF) of the relative distance $r(t)$ between TX and RX is given by \cite{absorb}
		\begin{align}
			\hspace{-0.7em}f_{r(t)}(r)\!=\!\dfrac{r}{r(0)\!\sqrt{2\!\pi \!D_\mathrm{TR}t}}\!\exp\!\left(\!-\dfrac{r^{2}\!+\!r(0)^{2}}{8\!D_\mathrm{TR}t}\!\right)\!\!\sinh\!\left(\!\dfrac{r(0)r}{4D_\mathrm{TR}t}\!\right). \label{pdfr}
		\end{align}
	The mobility of TX and RX dynamically modifies the channel. If the receiver is passive, the probability that a SM released from the TX at time $t$ is observed inside the RX after $\tau_{s}$ time from the transmission is expressed as \cite{stochastic}, \cite{rskmain}
		\begin{align}
			\hspace{-0.6em}h_{\mathrm{P}}(t, \tau_{s})=\dfrac{V_{\text{RX}}}{(4\pi (D_{\text{inf}}+D_\mathrm{TR})\tau_{s})^{3/2}} ~e^{-\dfrac{r(t)^{2}}{4(D_{\text{inf}}+D_\mathrm{TR})\tau_{s}}}, \label{cir}
		\end{align}
		where $V_{\text{RX}}$ is the volume of the spherical RX. 
		Moreover, for an absorbing RX, the probability of absorption of SMs during bit duration $T_{b}$ after transmission at time $t$ is given as \cite{absorb}
		\begin{align}
			h_{\mathrm{A}}(t, T_{b})=\dfrac{r_{\text{RX}}}{r(t)} ~\erfc\left(\dfrac{r(t)-r_{\text{RX}}}{2\sqrt{(D_{\text{inf}}+D_\mathrm{TR})T_{b}}}\right). \label{CIR}
		\end{align}
		Let $h$ denote the probability that molecule is detected at the receiver. Note that under Brownian motion of molecules and TX/RX devices, $h$ is equal to $h_{\mathrm{P}}(t, \tau_{s})$ and $h_{\mathrm{A}}(t, T_{b})$ as given in (\ref{cir}) and (\ref{CIR}) for passive and active receivers respectively. However, we consider a general setup that includes any arbitrary propagation mechanism with the constraint that the motion of each molecule is independent of each other, which is a standard assumption in MC under various propagation
		methods including Brownian motion.
	\section{Multi-Axis Concentration Modulation}
	In this work, we propose a generalized modulation framework termed multi-axis concentration modulation (MAxCM), utilizing a constellation space $C$ with multiple signal axes. In particular, we consider a $K$ dimensional constellation space where $i^{th}$ axis $\hat{s_{i}}$ represents the signal denoting concentration (or equivalently number of molecules) of $i^{th}$ type of SMs. Due to reception assumptions that the receiver can distinguish among these $K$ types of molecules, the signal axes are orthogonal to each other. Here, the information is encoded into joint concentrations of all types of molecules. As we clarify later, we can equivalently see that the information is encoded into the sum concentration and the ratio of concentrations of various types, analogous to amplitude and phase, respectively, in QAM modulation used in an EM based communication. Let $M$ denote the modulation order. Hence, MAxCM $\mathcal{M}(K,M)$ consists of $M$ constellation symbols where each constellation symbol $b$ is denoted by its coordinate point in the multi-dimensional constellation space $C$. Each symbol $b$ is represented as
	$	\mathbf{a}_b = (a_{1b}, a_{2b}, \dots, a_{Kb}),$
	where $a_{ib}$ is the number of molecules of type $i$ used for symbol $b$ for $b\in\{0,\dots, M-1\}$. Note that unlike conventional wireless communication, where the constellation spans all four quadrants of the in-phase and quadrature axis, molecular concentrations cannot take negative
	values. Therefore, all molecular constellations are restricted to the first quadrant.

Generally speaking, two distinct classes of constellations can emerge for MAxCM, along with a hybrid of the two and many variants. First class corresponds to rectangular lattice constellations which appear when encoding is carried out in terms of independent concentration levels for each
type of molecule. This results in the cartesian product of the constellation points of each axis resulting in a rectangular constellation similar to QAM in conventional communication. The second class corresponds to ratio-based constellations that encode the information in the relative concentration, i.e., the ratio of concentrations of molecules of different  types. These constellations maintain a constant total number of molecules per symbol across all
constellation points, producing
constellation symbols over a hyperplane in the
first quadrant (i.e. a line in 2D). 

We now show some specific examples of MAxCM and highlight their relation with existing modulation methods, wherever applicable.
\begin{itemize}[leftmargin=0.3cm]
\item \textbf{\textit{Single-axis CM, i.e., $\bm{K=1}$} of order $\bm{M^2}$}:
For single molecular type, the MAxCM becomes CSK (See Fig. \ref{dr}(a)), where a different concentration (i.e. number of molecules) is released for different symbols.
\item \textbf{\textit{MAxCM with constellation symbols located on the axes}}:
Consider a MAxCM($K,K$) of modulation order $K$ where $b$$^\text{th}$ constellation symbol is $a_{b}= ce_{b}$, where $e_{b}$'s denote unit vector in the direction of $b$$^\text{th}$ axis. This modulation is also known as molecular shift keying (MoSK).
\begin{figure}[h]
	\centering
	\includegraphics[width=\linewidth]{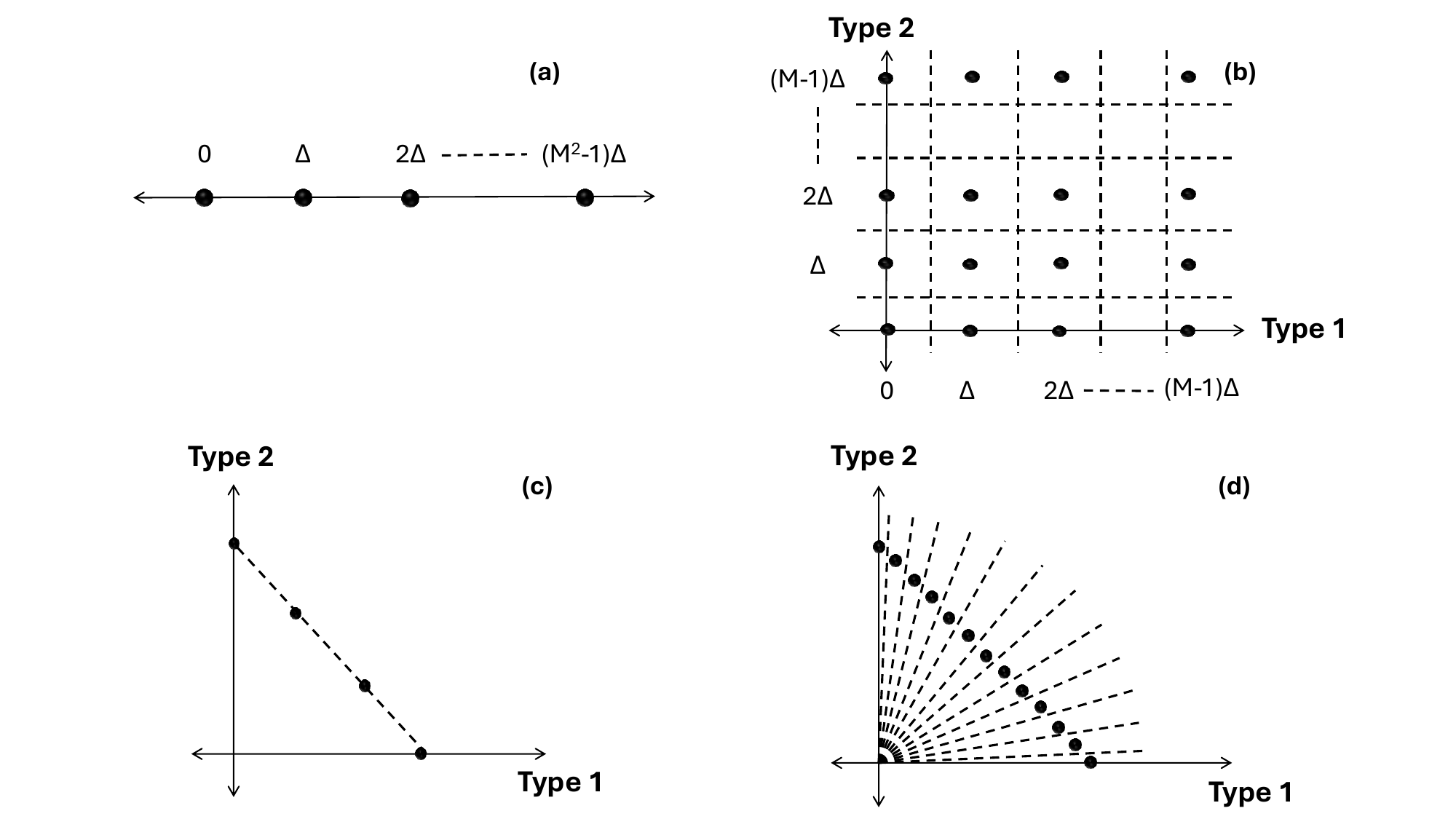}
	\caption{An illustration of two distinct classes of MAxCM constellation. First class corresponds to rectangular lattice constellation (first row) as (a) Single axis CM($K=1,M^2$) and (b) MAxCM(2,16). Second class corresponds to ratio encoding in line constellation (second row) as (c) MAxRSK(2,4) with symbols at $(c, 0), (2c/3, c/3), (c/3, 2c/3), (0, c)$,  and (d) MAxRSK(2,16). Ratios can be seen as analogous to phase in conventional QPSK constellation.}
	\vspace{-1em}
	\label{dr}
\end{figure}
\item \textbf{\textit{MAxCM with constellation symbols at different concentrations and ratios of two types of molecules}}: Fig. \ref{dr}(b) shows a variant of MAxCM (rectangular lattice constellation) with $K=2$ and $M^2=16$. This represents a modulation where the information is encoded in both total concentration and ratio of two types of molecules. Hence, it can be seen as analogous of QAM (total concentration as amplitude) in conventional communication. This variant is a novel modulation which has not been presented in any prior MC work.
\item \textbf{\textit{MAxCM with constellation symbols at different ratios of two types of molecules}}:
Let us consider a MAxCM($2,M$) with $K=2$, where $b$$^\text{th}$ constellation symbol is given as 
\begin{equation}\label{eq:2d-const}
	\mathbf{a}_{b}= \frac{c}{1+r_{b}}\,(1,\,r_{b})
	\rightarrow
	a_{1,b}=\frac{c}{1+r_{b}}, a_{2,{b}}=\frac{c r_{b}}{1+r_{b}},
\end{equation}
where $r_{b}$ denotes the ratio of number of molecules of type 2 with respect to number of molecules of type 1 for $b$$^\text{th}$ symbol. The total number of molecules is fixed at $c$, i.e., \(\,a_{1,b}+a_{2,b}=c\), and information is encoded only in the ratio \(r_{b}\) of two molecules, hence it is an example of ratio shift keying (RSK) (see Fig. \ref{dr}(c) and Fig. \ref{dr}(d)). The average concentration is $\bar{c}= c$. Such constellation can be seen as analogous of a QPSK constellation (ratio as phase). As explained earlier, all constellation points lie on a line due to the sum constraint on the transmitted molecules. 
\item 
The previous case can be extended to a higher value of $K$, by defining ratios with respect to first molecule type. For the $b^{\mathrm{th}}$ $K$-dimensional symbol, we define $K-1$ ratios $(r_{b2},r_{b3},\dots,r_{bK})$ with respect to first type. The symbol can be represented by the vector
\begin{equation}\label{eq:kd-const}
	\mathbf{a}_{b} \;=\; \frac{c}{1+\sum_{i=2}^{K} r_{bi}}\bigl(1,\,r_{b2},\,r_{b3},\dots,r_{bK}\bigr),
\end{equation}
which satisfies \(\sum_{i=1}^K a_{i,b}=c\). We term this \emph{multi-axis ratio shift keying} (MAxRSK). Here, the direction of the molecule concentration vector carries the information, not its magnitude.
As we will show later, during demodulation, the receiver estimates the transmitted symbol from the ratio of received number of SMs corresponding to each type. Since it utilizes ratios to decode transmitted symbols, it can be designed to be invariant to channel variation, hence not requiring any knowledge of channel which is one prime advantage of MAxRSK.
\end{itemize}
\vspace{-1em}
\subsection{Reception and Decoding Performance}
In an additive white Gaussian noise (AWGN) channel, the performance of  a particular modulation or constellation can be determined by its average power and separation between symbols. While these also serve as good indicators for performance in MC, they do not completely characterize the performance. 
Next, we present the decoding and BER analysis of MAxCM constellation to provide insights for constellation design. 
The received number of molecules of $i^{\text{th}}$ type at the $i^{\text{th}}$ RXP is given as $m_{i}=s_{i}+n_{i}$, where $s_{i}$ denotes the number of molecules of $i^{\text{th}}$ type intended for transmitted symbol $b$ and $n_{i}$ is the noise. Here, $s_{i}$ and $n_{i}$ are random variables distributed as \cite{diffMC} 
\begin{align}
	\hspace{-0.8em} ~s_{i} \sim \text{Binomial}(a_{ib},h), ~ n_{i} \sim \text{Poisson}(\lambda), \label{rm}
\end{align}
where $a_{ib}$ is the number of $i^{\text{th}}$ type molecules transmitted corresponding to symbol $b$, $h$ is the channel denoting the probability that a transmitted molecule is observed at the RX, and $\lambda$ is the mean number of noise molecules. Let $\mathbf{m}$ denotes the stacked received vector, i.e., $\mathbf{m}=(m_{1},m_{2}, \cdots, m_{K})$. From (\ref{rm}), the probability mass function (PMF) of the received signal $m_{i}$ is given as
\begin{align}
	\hspace{-0.6em}P_{m_{i}}(m_{i}|b)\!\!=\!\!\sum_{k=0}^{\min(a_{ib},m_{i})}\!\binom{a_{ib}}{k}\!\dfrac{\lambda^{m_{i}-k}}{(m_{i}-k)!}\!h^{k}\!(1-h)^{a_{ib}-k}\!e^{-\lambda}. \label{bin}
\end{align}
Note that the upper limit in the summation is due to the fact that the maximum number of intended SMs ($s_{i}$) received in a given bit interval can not be more than the transmitted number of molecules ($a_{ib}$). Since the reception at $K$ ports is independent, hence, the joint PMF of received molecules ($m_{1}, m_{2}, \cdots, m_{K}$) is
\begin{equation}
	P_{\mathbf{m}}(m_{1},m_{2}, \cdots, m_{K}|b)= \prod\nolimits_i P_{m_{i}}(m_{i}|b)\label{jointPMF}
\end{equation}
The ML decoding region for symbol $b$ is given by
\begin{align}
\mathsf{D}(b) =\left\{\mathbf{m} : P(\mathbf{m}|b) \geq P(\mathbf{m}|b') ~~\forall ~b' ~~\right\}. \label{dbb}
\end{align}
The above rule forms a partition of the whole constellation space. Note that for each symbol, its decoding region $\mathsf{D}(b)$ is the set of received vectors that are decoded as the symbol $b$.
\vspace{-1em}
\subsection{Decoding under Poisson Approximation}
In practical cases, $a_{ib}$'s are usually large and $h$ is small, therefore $s_{i}$ can be assumed to be distributed as a Poisson random variable with parameter $a_{ib}h$. Hence, we can write
\begin{align}
	&s_{i} \sim \text{Poisson}(a_{ib}h);~ n_{i} \sim \text{Poisson}(\lambda), \nonumber \\
	&\implies m_{i} \sim \text{Poisson}(a_{ib}h + \lambda). \label{rm2}
\end{align}
(\ref{rm2}) can also correspond to another practical scenario where the release of molecules from the TX is not perfect, leading to a random number $N$ of molecules being released, with $N$ being Poisson distributed.
Hence, the distribution of $\mathbf{m}$ is given as
	\begin{align}
		P(\mathbf{m}| b)=\prod_{i} e^{-(a_{ib}h+\lambda)} (a_{ib}h+\lambda)^{m_{i}}/m_{i}!, \label{pmb}
	\end{align}
where $\mu_{ib}=a_{ib}h+\lambda$ is the mean number of molecules received at RXP $i$ for transmit bit $b$. Now, we derive the decoding region of symbol $b$ as given in the following theorem:
\begin{theorem}
The ML decoding region for the symbol $b$ is
\begin{align}
	\mathsf{D}(b)=\{\mathbf{m}:(c_{b'}-c_{b})h \geq \sum_{i} m_{i} \log (\mu_{ib'}/\mu_{ib}) ~\forall~ b'\},\label{db}
\end{align}
where $c_{b'}= \sum_{i} a_{ib'}$ denotes the total number of molecules for symbol $b'$.
\end{theorem}
\begin{proof}
From (\ref{pmb}), the log-likelihood function corresponding to the symbol $b$ is 
\begin{equation}
	\ell(b)
	= \sum\nolimits_{i=1}^K \left( -\mu_{ib} + m_i \ln(\mu_{ib}) \right).
\end{equation}
Hence, the difference between log-likelihood of $b$ and $b'$ is
\begin{align}
	\ell(b) - \ell(b')= \sum\nolimits_{i=1}^K
	\left( -\mu_{ib} + \mu_{ib'} + m_i \ln\!\left(\mu_{ib}/\mu_{ib'}\right) \right).
\end{align}
The ML rule selects the symbol $b$ if
$\ell(b) - \ell(b') \ge 0$ for all $b' \neq b$. Hence, the decoding region for symbol $b$ is given by
\begin{small}
\begin{equation}
	\mathsf{D}(b)
	= \left\{
	\mathbf{m} :
	\sum_{i=1}^{K} (\mu_{ib'} - \mu_{ib})
	+ \sum_{i=1}^{K} m_i \ln\!\left(\frac{\mu_{ib}}{\mu_{ib'}}\right)
	\ge 0,
	\ \forall b' \neq b
	\right\}. \label{ca}
\end{equation}
\end{small}
Substituting $\mu_{ib} = h a_{ib} + \lambda$ and utilizing $c_b = \sum_{i=1}^{K} a_{ib}$, we note that
\begin{align}
	\sum_{i=1}^{K} \mu_{ib} = h c_b + K\lambda,\;\sum_{i=1}^{K} (\mu_{ib'} - \mu_{ib}) = h(c_{b'} - c_b).\label{iden}
\end{align}
Using (\ref{iden}), (\ref{ca}) can be equivalently written as
\begin{small}
\begin{equation}
	\mathsf{D}(b)\!=\! \left\{\mathbf{m}\! :h(c_{b'} - c_b)\!
	+\! \sum_{i=1}^{K} m_i\! \ln\!\left(\frac{\mu_{ib}}{\mu_{ib'}}\right)\!
	\ge 0,\!\!\ \forall \;b' \neq b
	\right\}. \label{ew}
\end{equation}
\end{small}
Rearranging the terms in (\ref{ew}) yields (\ref{db}).
\end{proof}
\begin{remark}
The above decoder is equivalent to a decoder that estimates $b$ as 
	\begin{equation}
		\widehat{b}= \arg \min_{b'}~ c_{b'}- \dfrac{1}{h}\sum_{i} m_{i} \log(\mu_{ib'}).
	\end{equation}
	Hence, it makes a decision based on the weighted L1 norm of the received vector subtracted from the total number of molecules. Hence, we term it as optimal weighted combining (OWC) decoder.
\end{remark}
\begin{remark}
The boundary between the decoding regions of two symbols $b$ and $b'$ is a hyperplane $\sum_{i} m_{i} \log (\mu_{ib}/\mu_{ib'}) +(c_{b'}-c_{b})h = 0$. Hence, the decision region associated with each constellation point is formed by the intersection of $K-1$ hyperplanes in the $K$-dimensional constellation space. 
\end{remark}
\begin{corollary} 
In MAxRSK, where $c_{b} =c$ and symbol $b$ is represented by $(cr_{1b},cr_{2b}, \cdots, cr_{Kb})$, with ratio $r_{ib}$ satisfying the constraint $\sum_{i} r_{ib}=1$, the decoding region for symbol $b$ can be simplified as
$\mathsf{D}(b)= \cap_{b'\ne b } \{ \mathbf{m}:  \sum_{i} m_{i} \log\left( (r_{ib}h+\lambda/c)/(r_{ib'}h+\lambda/c)\right) \geq 0\}$.
\end{corollary}
\begin{corollary}
For rectangular lattice constellation with independent encoding where $i^{\text{th}}$ axis encodes a bit (or a set of bits), the decoding can be done independently for each bit (or each set of bits). Here, the decoding region for the bit 1 is given as
\begin{align}
	\hspace{-0.5em}\mathsf{D}(1)\!=\!\left \{m_{i}: m_{i}>\dfrac{(A_\mathrm{1}-A_\mathrm{0})h(t)}{\ln{((A_\mathrm{1}h(t)+\lambda)/( A_\mathrm{0}h(t)+\lambda))}}\right\}, \label{dreq}
\end{align}
	where $A_\mathrm{1}$ and $A_\mathrm{0}$ number of molecules are transmitted for bit $1$ and $0$ respectively. 
\end{corollary}
The probability of symbol error is given by 
\vspace{-1em}
\begin{equation}
	P_s \;=\; 1 - \frac{1}{\mathcal{M}}~\sum_{b\in\mathcal{C}}~
	\sum_{\mathbf{m}\in\mathsf{D}(b)} P(\mathbf{m}\mid b),
	\vspace{-1em}
\end{equation}
where $\mathcal{M}$ is the set of symbols.

It can be observed from Theorem 1 that decoding in concentration-based modulation (including MAxCM, CSK etc.) requires CI, since concentration of molecules undergo attenuation during propagation in an environment-dependent manner. Estimating CI is often challenging and, in many practical MC scenarios, may even be infeasible. In MC systems operating in static environment, the channel does not vary over time. Hence, in such cases, channel estimation can be performed once and can be subsequently used for decoding. When the accurate theoretical models of the underlying propagation mechanisms and channel are known, it is possible to estimate the static channel via theoretical equations, provided that the underlying variables (including diffusion coefficients) are known. However, any mismatch between the known and actual values of these parameters (including modeling inaccuracies) can lead to a significant increase in decoding error probability. These difficulties become more pronounced in dynamic MC systems, where the channel varies with time. In such cases, periodic channel estimation is required, which can become extremely difficult and practically infeasible. 
Under these circumstances, the decoding strategy depends critically on whether the instantaneous channel gain $h$ is known or not. When $h$ is known, the optimal decoder can be employed utilizing (\ref{dbb}) and (\ref{db})), as described in this section. Notably, in MMC systems, this may result in a time-varying decoder that adapts to temporal variations in $h$. However, when $h$ is not known, decoding must rely on the CSI, which may consist of the complete distribution of $h(t)$ or partial knowledge of $h(t)$, such as its mean value, as discussed next.
\begin{table*}[!t]
	\centering
	\caption{Comparison Between Single-Axis and Multi-Axis rectangular MAxCM}
	\label{tabcom}
	\begin{tabular}{|c|c|c|c|c|c|}
		\hline
		\makecell{Order \\ $M^2$} & \makecell{Bits/symbol \\ $2N$} & 
		\makecell{Multi axis MAxCM(2,$M^2$) \\ Sym. sep. ($c'_{\mathrm{MA}}$)} &
		\makecell{Single axis MAxCM(1,$M^2$) \\ Sym. sep. ($c'_{\mathrm{SA}}$)} &  \makecell{Per-symbol \\ Budget}  & \makecell{Per-bit \\ Budget}    \\
		\hline
		4  & 2 & $2c$ & $\frac{4}{3}c$& $2c$  & $c$   \\
		
		16 & 4 &  $\frac{4}{3}c$&$\frac{8}{15}c$ &$4c$  & $c$   \\
		
		 64 &  6 & $\frac{6}{7}c$& $\frac{4}{21}c$ &$6c$  & $c$   \\
		\hline
	\end{tabular}
\end{table*}
\begin{enumerate}[leftmargin=0.5cm]
\item \textbf{When distribution of the channel $h$ is known:}
ML decoder selects the symbol that maximizes the mean likelihood, i.e.,
\begin{align}
\hat{b} = &\arg \max_{b} \mathbb{E}_h[P(\mathbf{m}\mid b,h)],\label{logp}
\end{align} 
where $P(\mathbf{m}\mid b,h)$ is given in (\ref{pmb}). 
Equivalently, we can also maximize the log likelihood as
\begin{equation}
	\hat{b} = \arg\max_{b} \log \mathbb{E}_h[P(\mathbf{m}\mid b,h)]. \label{19}
\end{equation}
\item \textbf{When partial information about the distribution of $h$ is known}: Since $\log(\cdot)$ is concave, from Jensen's inequality, we can get
\begin{equation}
	\log \mathbb{E}_h[P(\mathbf{m}\mid b,h)] 
	\geq 
	\mathbb{E}_h[\log P(\mathbf{m}\mid b,h)].
\end{equation}
Hence, maximizing $\mathbb{E}_h[\log P(\mathbf{m}\mid b,h)]$ provides a computationally tractable alternative for the exact ML rule. Consequently, using (\ref{pmb}), we get
\begin{align}
	&\mathbb{E}_h[\log P(\mathbf{m}\mid b,h)] \nonumber\\
	\hspace{-5em}&= \sum\nolimits_{i=1}^{K} \mathbb{E}_h\!\left[-(h a_{i b} + \lambda) + m_i \log(h a_{i b} + \lambda) - \log m_i!\right] \notag\\
	&= -\sum\nolimits_{i=1}^{K} (a_{i b}\,\mathbb{E}[h] + \lambda)
	+ \sum\nolimits_{i=1}^{K} m_i\,\mathbb{E}_h[\log(h a_{i b} + \lambda)]\nonumber\\
	& \hspace{2em}-\sum\nolimits_{i=1}^{K} \log m_i!.
\end{align}
Note that $\sum_{i=1}^{K} \log m_i!$ term can be dropped as it does not depend on $b$. Hence, the approximate expected-log ML decoder can be written as
\begin{align}
	\hat{b}_{\text{Elog}}
	= \arg\max_b 
	&\bigg[-\sum\nolimits_{i=1}^{K} a_{i b}\,\mathbb{E}[h]
	 \nonumber\\
	&+\sum\nolimits_{i=1}^{K} m_i\,\mathbb{E}_h[\log(h a_{i b} + \lambda)]\bigg].
\end{align}
This is an explicit decision metric that uses only the moments $\mathbb{E}[h]$ and the expectation 
$\mathbb{E}_h[\log(\cdot))]$; importantly, it does not require specifying the full PDF of $h$, if these expectations can be computed or approximated. 
\item \textbf{When only the mean of channel is known}: When mean of $h$ is known, a suboptimal decoder can determined by substituting $\mathbb{E}[h(t)]$ in place of $h(t)$ in (\ref{pmb}).
\end{enumerate}
This motivates us to propose optimal constellation design that
allows decoding without any channel knowledge. In MAxRSK,
where information is encoded in the ratios of concentrations,
it may be possible to design constellation such that decoding
does not require any CI. This is due to the ratio of concentrations of different types being preserved
by the channel, as the
type identity of molecules does not change during propagation.
This is in sharp contrast to wireless communications, where
constellation rotation errors can occur under fading. There are various benefits of MAxCM. First, in MAxCM, multiple axes provides multiple orthogonal directions allowing the use of higher modulation order and hence increasing SE. Further, MAxRSK (a variant
of MAxCM) can be designed to be invariant of channel variation
and hence, can provide superior performance in dynamic
channel conditions or the scenario where CI is
not present. It is evident that the proposed modulation MAxCM
provides a general framework to design various modulation
methods that can be suitable under specific scenario. In the next two sections, we explore two subclasses of MAxCM to demonstrate the above. First, we will present the rectangular MAxCM analogous to QAM and show how the use of multiple axes can provide better BER and SE. Next, we describe a special sub-class SMAxRSK constellation where
channel-independent decoding is possible.
\vspace{-1em}
\section{Rectangular MAxCM}
We now explore the rectangular MAxCM $(K,M)$ for MC, analogous to QAM modulation
in wireless systems, where information is encoded in both concentration and ratio of $K$ types of molecules and
discuss the gain that can be achieved by the use of  multiple axes. For better exposition of the idea, let us consider $K=2$. In particular, we consider a $M \times M$ MAxCM$(2,M^2)$, as shown in Fig. \ref{dr}(b), with $(i,j)$ symbol
at $(i c', j c')$, where $i,j \in \{0,1,\ldots,M-1\}$. This constellation corresponds to $2N$ bits per symbol where $N = \log_2 M$. We consider a budget of $c$ molecules per bit or equivalently $2Nc$ molecules per symbol
The per-symbol power of this constellation is given by
\begin{equation}
	\frac{1}{M^2} \sum_{i=0}^{M-1} \sum_{j=0}^{M-1} (i c' + j c')
	= (M-1)c'.
\end{equation}
Hence, under the same molecular budget, we have 
\begin{equation}
	c'_{\mathrm{MA}} = \frac{2Nc}{M-1}=\frac{2 c \log_2 M}{M-1}.
\end{equation}
To understand the benefit of using multiple axes, let us consider single-axis CM
MAxCM$(1,M^2)$ of the same order $M^2$ (as shown in Fig. \ref{dr}(a)).
Its per-symbol power is
\begin{equation}
	\frac{1}{M^2} \sum_{i=0}^{M^2-1} i c'
	= \frac{1}{M^2} \cdot \frac{(M^2-1)M^2}{2} c'
	= \frac{M^2-1}{2} c'.
\end{equation}
Hence, keeping the same budget, we get
\begin{equation}
	c'_{\mathrm{SA}} = \frac{4Nc}{M^2-1}=\frac{4c\log_2 M}{M^2-1}.
\end{equation}
Table \ref{tabcom} presents the comparison of single-axis and multi-axis MAxCM under equal per-bit molecular budget.
From the comparison, we observe two important insights:
\begin{enumerate}
	\item MAxCM$(2,M^2)$ can provide better symbol separation than 
	MAxCM$(1,M^2)$ for the same SE and molecular budget.
	This shows that the use of multiple axes can provide performance gains in
	symbol error rate (SER).
	\item MAxCM$(2,M^2)$ can achieve double the SE as compared to MAxCM$(1,M)$ while maintaining equal per-bit molecular budget
	and similar symbol separation. This demonstrates that employing multiple
	axes can improve SE without degrading SER performance.
\end{enumerate}
\vspace{-1em}
	\section{Symmetric MAxRSK}
	We now design and demonstrate a special case of MAxRSK where symbol decoding does not require any CI. The ability to design such modulation schemes highlights the versatility and strength of the proposed MAxCM and MAxRSK frameworks. To provide a clear exposition, we first focus on a simple case of MAxRSK with modulation order $M=2$ and two molecule types ($K=2$), which we term as BRSK, as shown in Fig. \ref{cons}.
	Here, the TX uses two different types of molecules for information transmission to the RX, i.e., A (Type 1) and B (Type 2). The intended information is encoded in a parameter $\rho \in [0,1] $, defined as the ratio of the concentration of the two types of molecules at the TX. In particular, for modulation order 2, the bits 0 and 1 are encoded using ratios $\rho_{0}$ and $\rho_{1}$, respectively.
	Hence, for bit 0, $a_{10}=\rho_{0}c$ number of Type 1 molecules and $a_{20}=(1-\rho_{0})c$ number of Type 2 molecules are transmitted from the TX, where $c$ is the molecular budget, i.e., the total number of molecules allocated for each transmitted bit. Similarly, for bit 1, $a_{11}=\rho_{1}c$ number of Type 1 molecules and $a_{21}=(1-\rho_{1})c$ number of Type 2 molecules are transmitted. 
	As we show later, the optimal receiver estimates the symbol from the ratio of received number of molecules corresponding to each type, i.e., $m_{1}/ m_{2}$. 	 
We will first consider the static scenario in next subsection which we later extend to dynamic scenario.
\vspace{-1em}
	\subsection{Static MC}
	In a static MC environment with immobile TX and RX (i.e., $D_{\text{RX}}=D_{\text{TX}}=D_\mathrm{TR}=0$), only the messenger molecules undergo Brownian motion. As a result, the propagation channel remains time-invariant over the entire observation interval. For simplicity, we have considered the equal probability of transmission of bit 0 and bit 1.
	\subsubsection{Optimal ML Decoder}
	From ($\ref{bin}$) and ($\ref{jointPMF}$), the joint PMF of $m_{1}$ and $m_{2}$ is
\begin{small}
\begin{align}
	\hspace{-0.65em}P(m_{1},m_{2}|b)\!=\! \prod_{i=1}^{2}\!\sum_{k=0}^{\min(a_{ib},m_{i})}\! \binom{a_{ib}}{k}\dfrac{ e^{-\lambda}\!\lambda^{m_{i}-k}}{(m_{i}\!-\!k)!}h^{k}(1\!-\!h)^{a_{ib}-k}. \label{oml}
\end{align}
\end{small}
	Substituting (\ref{oml}) in (\ref{dbb}), we get the ML decoding region for bit 1 as
		{\footnotesize
			\begin{align}
				&\mathsf{D}(1)=\Bigl\{(m_{1},m_{2}): \sum_{k_{1}=0}^{\min(a_{11},m_{1})} \sum_{k_{2}=0}^{\min(a_{21},m_{2})} \binom{a_{11}}{k_{1}} \binom{a_{21}}{k_{2}} \nonumber \\
				&\binom{m_{1}}{k_{1}} \binom{m_{2}}{k_{2}} k_{1}! k_{2}! \left(\dfrac{h}{(1-h)\lambda}\right)^{k_{1}+k_{2}} \geq \sum_{k_{3}=0}^{\min(a_{10},m_{1})} \sum_{k_{4}=0}^{\min(a_{20},m_{2})} \nonumber \\
				& \binom{a_{10}}{k_{3}} \binom{a_{20}}{k_{4}} \binom{m_{1}}{k_{3}} \binom{m_{2}}{k_{4}} k_{3}! k_{4}!\left(\dfrac{h}{(1-h)\lambda}\right)^{k_{3}+k_{4}} \Bigr\}. \label{bin2}
			\end{align}
		}
		The optimal decoding rule is hence given as
	\begin{small}
			\begin{align}
				\widehat{b}= \begin{cases}
					1, &  \text{if}\;(m_{1},m_{2})\in \mathsf{D}(1).\\
					0, & \text{otherwise}. \label{cases}
				\end{cases}
			\end{align}
		\end{small}
		\subsubsection{Poisson Approximation}
		Under the Poisson approximation for $m$'s distribution, as described in Section IIIB, from (\ref{pmb}), the joint PMF of $m_{1}$ and $m_{2}$ is
		\begin{align}
			P(m_{1},m_{2}| b)=\prod_{i=1}^{2} e^{-(a_{ib}h+\lambda)} (a_{ib}h+\lambda)^{m_{i}}/m_{i}!. \label{pmb2}
		\end{align}
		Using the result derived in Theorem 1 and setting $K=2$ with $c_{b'}=c_{b}=c$, the decoding region $\mathsf{D}(1)$ reduces to
		\begin{align}
			\hspace{-0.5em}m_{1} \ln{\dfrac{\mu_{11}}{\mu_{10}}} + m_{2} \ln{\dfrac{\mu_{21}}{\mu_{20}}} \geq 0 
			\implies \sum_{i=1}^{2} \ln \left(\dfrac{\mu_{i1}}{\mu_{i0}}\right)^{m_{i}} \geq 0. \label{d2}
		\end{align}
		(\ref{d2}) shows that the optimal ML decoder for BRSK-based static MC system under Poisson reception is a weighted combiner with weights dependent on the mean received signal strength at both RXPs. Putting $\mu_{ib}=a_{ib}h+\lambda$ and $\mathbf{a}_{b}=(\rho_{b}c, (1-\rho_{b})c)$, we get the following decoding region for bit 1 (assuming $\rho_{1} > \rho_{0}$)
		\begin{align}
		&m_{1} \ln{\dfrac{c\rho_{1}h+\lambda}{c\rho_{0}h+\lambda}} + m_{2} \ln{\dfrac{c(1-\rho_{1})h+\lambda}{c(1-\rho_{0})h+\lambda}} \geq 0, \nonumber \\
		&\implies\dfrac{m_{1}}{m_{2}}\geq -\dfrac{\ln \left(\dfrac{c(1-\rho_{1})h+\lambda}{c(1-\rho_{0})h+\lambda}\right)}{\ln\left(\dfrac{c\rho_{1}h+\lambda}{c\rho_{0}h+\lambda}\right)} \stackrel{\Delta}{=} \eta.
		\label{main}
		\end{align}
		Note that (\ref{main}) corresponds to a ratio-based decoder which means that the ratio $m_{1}/m_{2}$ is compared with a threshold $\eta$ to decode the transmitted bit. If the ratio is above the threshold, bit 1 is decoded, else bit 0 is decoded. The decoder is illustrated by a straight line in Fig \ref{cons}.
		\begin{figure}[h]
			\centering
			\includegraphics[trim={50mm 30mm 50mm 30mm},clip,width=0.75\linewidth]{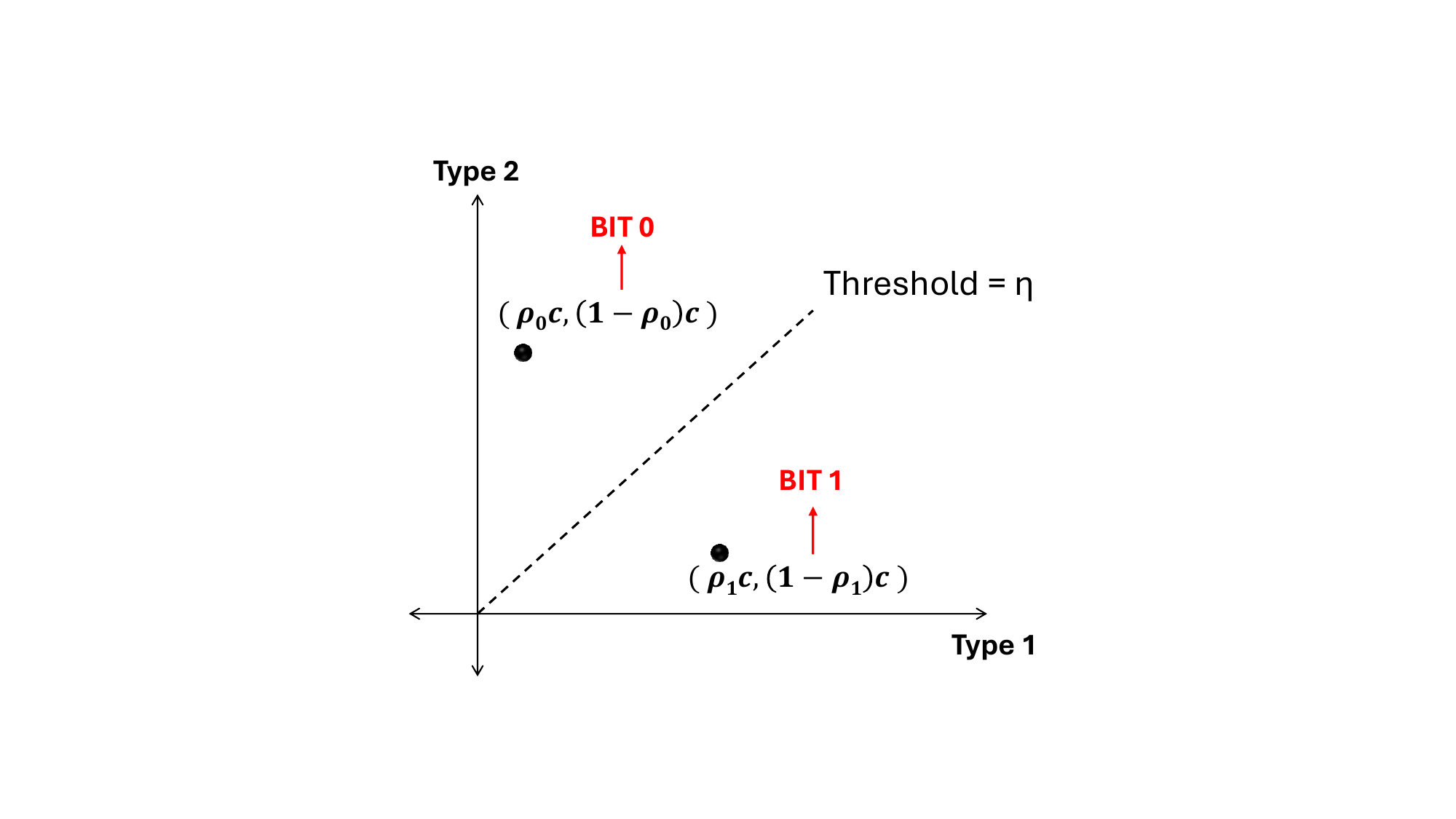}
			\vspace{-1em}
			\caption{An illustration of symbol constellation for BRSK along with decoding regions.}
			\vspace{-1em}
			\label{cons}
		\end{figure}
		
		While the ratio threshold $\eta$ generally depends on $h$, the dependency on $h$ can be removed by carefully designing $\rho_{0}$ and $\rho_{1}$. For example, let us consider $\rho_{0}+\rho_{1}=1$. This results in a symmetric BRSK (SBRSK) with symbols ($c\rho_{0}$, $c(1-\rho_{0})$) and ($c(1-\rho_{0})$,$c\rho_{0}$). Substituting $\rho_{0}=1-\rho_{1}$ in (\ref{main}) results in a cancellation of numerator and denominator, resulting in $\eta=1$. Hence, the decoder is simply a comparator given as
		\begin{align}
			\hat{b}= \begin{cases}
				1, &  \text{if}\;(m_{1} \geq m_{2}).\\
				0, & \text{otherwise}. \label{cases2}
			\end{cases}
		\end{align}
		It is very interesting to note that for the symmetric MAxRSK(2,2) case, the optimal decoder becomes independent of the channel condition, i.e., the value of $h$.

		Note that in OOK modulation, under Poisson assumption, the optimal decoding region for bit 1 is given as \cite{abhisheksir}
		\begin{align}
			\mathsf{D}(1)&=\left\{m: e^{-(N_\mathrm{a}h+\lambda)} (N_\mathrm{a}h+\lambda)^{m}/m! \geq e^{-\lambda} \lambda^{m}/m!\right\} \nonumber \\
			&=\left \{m: m>\dfrac{N_\mathrm{a}h}{\ln{((N_\mathrm{a}h+\lambda)/\lambda)}}\right \},
			\label{dook2}
		\end{align}
		where $N_\mathrm{a}$ is the total number of molecules transmitted for bit 1 in the OOK modultion.
		This implies that the threshold for detecting bit 1 in the OOK modulation scheme depends on the value of $h$. Therefore, the receiver requires the exact CI to correctly decode the bit.
			\subsubsection{Probability of Error}
			The probability of error for BRSK or MAxRSK(2,2) is given as
	\begin{align}
	\hspace{-0.7em}P_{\mathrm{e}}\!=\!\dfrac{1}{2}\!\left(P_{\mathrm{e}|b=0}+P_{\mathrm{e}|b=1}\right)\!, ~\text{where}\;\label{perr}
				    P_{\mathrm{e}|b}\!=\!\sum\nolimits P(m_{1},m_{2}|b). 
				    \end{align}
	\subsection{Extension to Higher Order}
Symmetric BRSK can be naturally extended to higher-order constellations while preserving
channel-independent decoding properties. To demonstrate this, we present the design approach for such extension along with an example.
Recall that for MAxRSK $(K,M)$, the symbol $b \in \{0,\ldots,M-1\}$ is given as $\mathbf{a}_b = c \mathbf{R}_b$, with $\mathbf{R}_b$
representing a normalized ratio vector $\mathbf{R}_b = (r_{1b}, r_{2b}, \ldots, r_{Kb})$ where
$	\sum_{i=1}^{K} r_{ib} = 1.$
Here, $c$ is the fixed molecular budget for all symbols.
From Cor. 1.1, we get the ML decoding rule for symbol $b$ as
\begin{equation}
	\sum_{i=1}^{K} m_i
	\log\!\left(
	\frac{c r_{ib} h + \lambda}{c r_{ib'} h + \lambda}
	\right) \ge 0,
	\qquad \forall~ b' \neq b.
	\label{eq:maxrsk_decoding}
\end{equation}
For channel-independent decoding, the log term for different $i$'s in
\eqref{eq:maxrsk_decoding} must be multiple of each other, to enable simultaneous cancellation of the log term from all summand terms. This means that the ratios $\frac{c r_{ib} h + \lambda}{c r_{ib'} h + \lambda}$ for different $i$'s can take one of the three possible values: they are 1, identical across different $i$'s, or the inverse of each other. This requirement can be satisfied if, for any pair of symbols
$b$ and $b'$, the corresponding ratio pair ($r_{ib}$, $r_{ib'}$) for different $i$'s must either be of equal proportion, or the same, or be swapped versions of one another. In other words, for each pair $v=(b,b')$, we have
\begin{align}
	(r_{ib}, r_{ib'})\in 
	\left\{
	(\alpha_v,\alpha_v), (\alpha_v,\beta_v), (\beta_v,\alpha_v)
	\right\},
	\ \text{for each}~ i,\label{eq:ext_app}
\end{align} 
along with the constraint $\sum_{i=1}^K r_{ib}=1$, where $\alpha_v$'s and $\beta_v$'s are arbitrary constants between 0 and 1.
Utilizing $(\ref{eq:ext_app})$, a symmetric
MAxRSK$(3,3)$ constellation
can be constructed with three molecular types and three symbols with channel invariant decoding.
Let $p \in (0,1/2)$ be a design parameter. The three constellation
symbols are defined as
\begin{align}
	\mathbf{a}_0 &= c~(p,\; 1-2p,\; p), \\
	\mathbf{a}_1 &= c~(1-2p,\; p,\; p), \\
	\mathbf{a}_2 &= c~(p,\; p,\; 1-2p).
\end{align}
Each symbol is a cyclic permutation of the other, and all satisfy
$\sum_i r_{ib} = 1$.
Assuming $p > 1-2p$, the decoding regions simplify to pairwise
comparisons of the received molecule counts as
\begin{align}
	\mathsf{D}(0) &: \quad m_1 - m_2 \ge 0,\;\; m_3 - m_2 \ge 0, \\
	\mathsf{D}(1) &: \quad m_2 - m_1 \ge 0,\;\; m_3 - m_1 \ge 0, \\
	\mathsf{D}(2) &: \quad m_1 - m_3 \ge 0,\;\; m_2 - m_3 \ge 0.
\end{align}
We can observe that these decision regions are independent of the channel gain $h$ and
depend only on the relative received molecule counts, thereby eliminating the need for CI.
The presented MAxRSK(3,3) scheme represents a fundamental extension of ratio-based molecular modulation beyond the binary case and constitutes the first non binary and multi-molecule RSK scheme that enables channel-independent decoding in the presence of noise.
The presented approach can also be used to construct further higher order SMAxRSK.
An illustration of the proposed SMAxRSK(3,3) modulation scheme, including its constellation symbols and decoding regions, is shown in Fig. \ref{extend}.
 \begin{figure}[h]
	\centering
	\includegraphics[width=\linewidth]{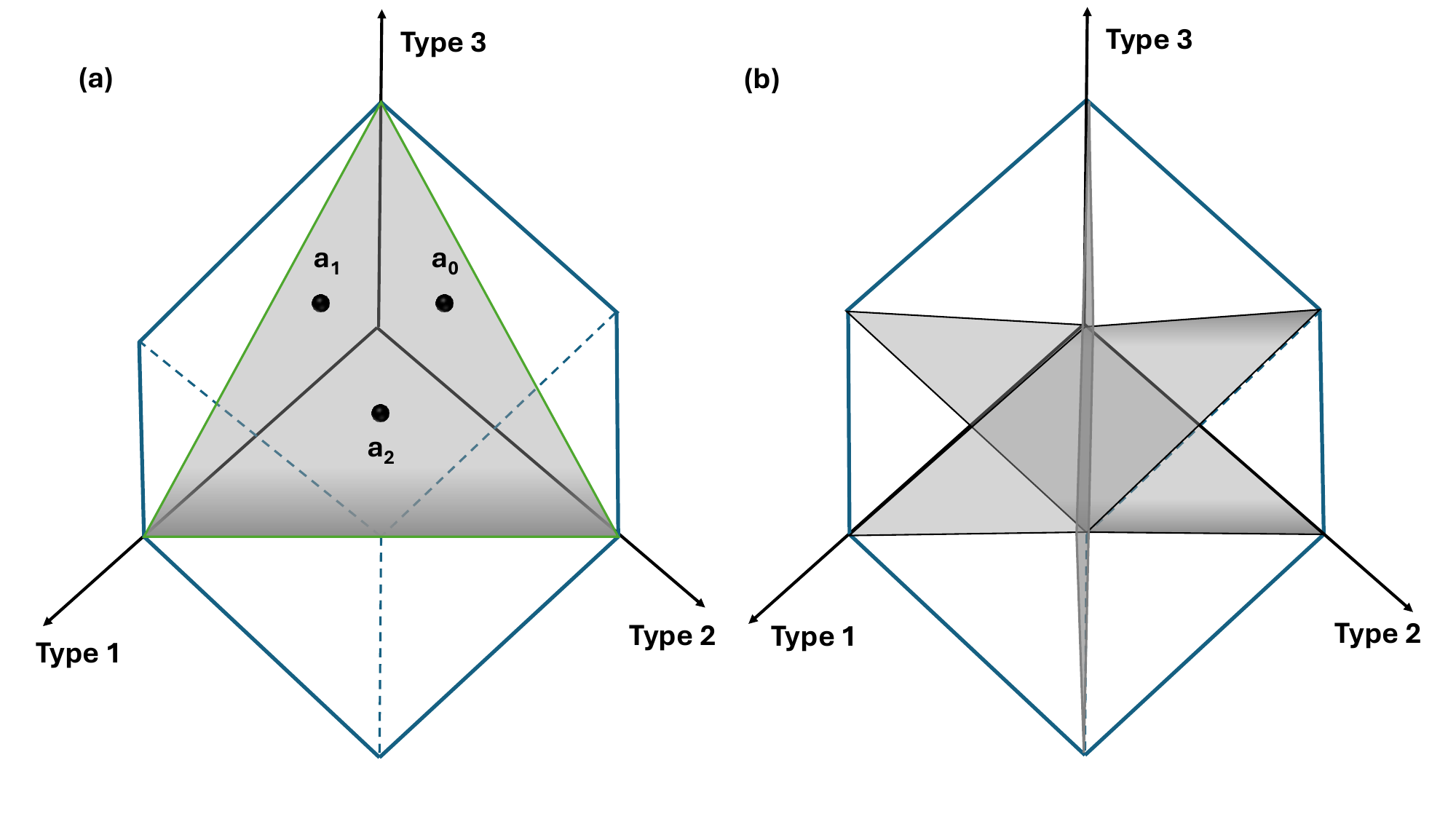}
	\vspace{-2em}
	\caption{An illustration of higher order ratio encoding as SMAxRSK(3,3) with (a) constellation symbols $\mathbf{a}_0$, $\mathbf{a}_1$, and $\mathbf{a}_2$ lying inside the plane $a_{1b}+a_{2b}+a_{3b}=c$ (shown with green boundary) and (b) decoding regions depicted as gray planes for each symbol.}
	\vspace{-1em}
	\label{extend}
\end{figure}
\subsubsection{Noise-free scenario}
Note from Cor. 1.1 that, in the absence of noise ($\lambda= 0$), the decoding law becomes independent of $h$ for all MAxRSK, regardless of the values of {$r_{ib}$}'s. This shows that, in the absence of noise, MAxRSK scheme does not require any CI, irrespective of its specific constellation design. Moreover, in cases where noise is small, MAxRSK can be considered robust to channel variations (at least in an approximate sense).
		\subsection{Dynamic MC}
		  We now study the performance of SMAxRSK in the dynamic MC environment where $h$ is varying over time. 
		As discussed in Section III, in such scenarios, decoding strategy generally depends on the fact whether $h$ is known or unknown. As shown earlier, the $h$'s distribution depends on propagation environment and receiver characteristics. For passive receiver, under diffusive propagation of SMs, and devices, the PDF of time-varying CIR $f_{h_{\mathrm{P}}(t)}(h)$, with $0 \leq h \leq \phi$, is given by \cite{stochastic}.
		Similarly, for the absorbing receiver, the channel PDF is given in \cite{absorb}.
		\subsubsection{SMAxRSK}
		 Since the decoder for SMAxRSK does not utilize any CI (cf. (\ref{eq:maxrsk_decoding})), the decoder remains the same even in dynamic MC environment and offers robustness in scenarios where CI/CSI is unavailable or impractical to obtain.
		The probability of error is given as
		$P_\mathrm{e}(t)= \mathbb{E}_{h(t)}[P_\mathrm{e}(h)],$
			where $h(t)$ is the value of $h$ at time $t$ and the error probability depends on time $t$. When noise is absent or significantly small, the channel invariant decoding extends to all MAxRSK.
			\subsubsection{OOK}
			For a comparative analysis, we also analyze OOK ($M=2$) in dynamic MC environment. 
		%
		In contrast, the decoder for OOK as given in (\ref{dook2}) requires the exact value or at least an estimate of the channel $h$. This increases the complexity significantly as the channel needs to be estimated regularly, as $h$ is time varying in a mobile MC environment. It is hence interesting to consider practical scenarios where exact channel is unknown or we can utilize partial information about the channel. We describe few such configurations below. 
		\paragraph{No channel information}
		In case, no estimate of $h$ is known, the OOK modulation becomes unusable while SMAxRSK modulation can be used efficiently.
		\paragraph{Mean channel information}
		When the mean of $h$ at time $t$ is known, the approximate estimate of threshold in (\ref{dook2}) can be computed by replacing $h$ by $\mathbb{E}[h]$. Hence, the decoding region for bit 1 is given as
		\begin{align}
			\mathsf{D}(1)= \left\{m: m>\dfrac{N_\mathrm{a}\mathop{\mathbb{E}}(h(t, \tau_{s}))}{\ln((N_\mathrm{a}\mathop{\mathbb{E}}(h(t, \tau_{s}))+\lambda)/\lambda)}\right\}, \label{dook}
		\end{align}
		which depends on the value of time $t$.
		
		For a system with a passive receiver, the mean of $h$ can be computed as \cite{stochastic}
		\begin{align}
		\hspace{-0.7em}	\mathop{\mathbb{E}}(h_{\mathrm{P}}(t, \tau_{s}))=\dfrac{V_{\text{RX}}}{\left(4\pi (D_{1}\tau_{s} + D_{2}t)\right)^{3/2}}e^{-\dfrac{r(0)^{2}}{4(D_{1}\tau_{s} + D_{2}t)}}, \label{meanhp}
		\end{align}
		where $r(0)$ is the initial distance between TX and RX.
		
		For an absorbing receiver, the information about the mean of channel is mathematically intractable \cite{absorb} and can only be computed numerically. Further, we can approximate the mean of the channel by replacing $r(t)$ in (\ref{CIR}) by $\mathbb{E}[r(t)]$, where $\mathbb{E}[r(t)]$ is the mean of the relative distance between the TX and the absorbing RX unit at time $t$. Using (\ref{pdfr}), the mean is given as \cite{absorb}
		\begin{align}
			&\mathop{\mathbb{E}}(r(t))=\sqrt{\dfrac{8D_\mathrm{TR}t}{\pi}}\exp(-\dfrac{r(0)^2}{8D_\mathrm{TR}t})+\left(r(0)+\dfrac{4D_\mathrm{TR}t}{r(0)}\right)\nonumber \\
			& \hspace{15em}\times\erf \left(\dfrac{r(0)}{8D_\mathrm{TR}t}\right). \label{meanr} 
		\end{align}
		
		\paragraph{Channel statistical information (CSI)}
		When an exact CI is not available, but the complete distribution of channel is known, we can use this distribution to determine a suboptimal threshold. Here, we determine the optimal threshold that minimizes the decoding error averaged over the channel distribution. Hence, the decoding region for OOK is given as:
		\begin{align}
			&\mathsf{D}_\text{OOK}(1)\!=\!\Bigl\{m: \mathbb{E}_h[P(m|b=1,h)] 
			\geq \mathbb{E}_h[P(m|b=0,h)]\Bigr\} \nonumber \\
			&=\biggl\{m: \int e^{-N_\mathrm{a}h} ((N_\mathrm{a}h+\lambda)/\lambda)^{m}  f_{h(t,\tau)}(h) ~ dh \geq 1 \biggr\}. \label{dook3}    
		\end{align}
	 
	 It is important to understand that in practical scenario, the devices may have inaccurate or partial information about CSI. Examples include cases when the information regarding the diffusion coefficient $D_{\text{inf}}$ of SMs is not accurate, or the parameters determining the movement of TX and RX devices may not be precisely known, or the actual propagation mechanism deviates from the theoretical one. In such cases of imperfect CSI, the estimate of threshold for OOK may not be correct and may increase the probability of bit error, as we observe in the numerical results.
	\begin{table}
	\centering
	\caption{Default values of different system parameters}
	\begin{tabular}{ | p{5.2cm} | p{2.5cm}| } 
		\hline
		Parameter & Value \\ \hline
		Initial distance between TX and RX ($r(0)$)	& 1 $\mu$m \cite{stochastic} \\ \hline
		Radius of the RX ($r_{\text{RX}}$) & 0.45 $\mu$m \\ \hline
		Diffusion coefficient of SMs ($D_{\text{inf}}$) & 5 $\times$ $10^{-9} \text{m}^{2}$/s \cite{stochastic} \\ \hline
		Diffusion coefficient of TX and RX ($D_\mathrm{TR}$) & 5 $\times$ $10^{-12} \text{m}^{2}$/s \\ \hline
		Bit duration ($T_{b}$) & 0.5 ms \cite{stochastic}\\ \hline
		Relative time of sampling ($\tau_{s}$) & 0.035 ms \cite{stochastic} \\ \hline
	\end{tabular}
	\label{tab2}
\end{table}
	\begin{figure}[t]
	\centering
	\includegraphics[width=0.95\linewidth]{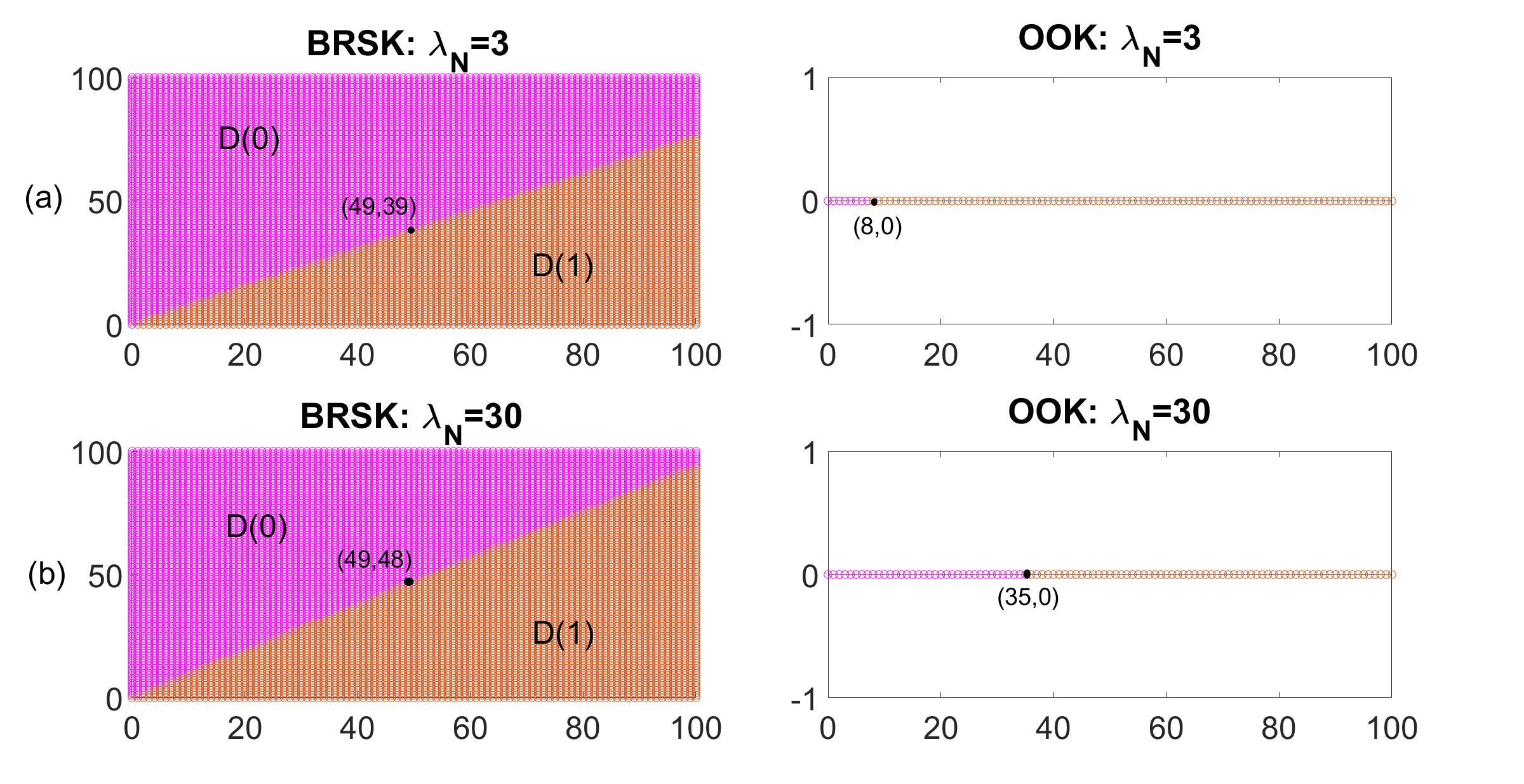}
	\vspace{-1em}
	\caption{Decoding region for $m_{1}$ and $m_{2}$ ranging from 0 to 100 for fixed $c=400$ and $h_{\text{LOW}}=0.0281$. Pink and orange region represent decoding region of bit 0 and 1.}
	\vspace{-1em}
	\label{hlow}
\end{figure}
\vspace{-1em}
		\section{Numerical Results}
		In this section, we present numerical results demonstrating the performance for static and mobile MC systems for BRSK modulation along with the results for the conventional benchmark OOK modulation scheme for a comparison. We provide the analysis for both passive and absorbing RX. For OOK, $N_{a}=2c$ number of molecules are transmitted for bit 1. The average number of molecules per symbol for both BRSK and OOK modulation schemes is kept the same, i.e., $c$ for a fair comparison.
		If not stated explicitly, the default values of system parameters are tabulated in Table \ref{tab2}. The number of simulation iterations is taken to be $10^{6}$. 
		\subsection{Static MC Environment}
		\subsubsection{Decoding regions for BRSK}
			To understand the decoding regions for BRSK, we consider an asymmetric BRSK with $\rho_{0}=0.4$, $\rho_{1}=0.8$, and present the decoding regions for bit 0 and bit 1 for fixed $c=400$ and $h_{\text{LOW}}=0.0281$. We can observe that the decision boundary is a line with its
		slope corresponding to the threshold $\eta$ given in (\ref{main}). It can be observed from (\ref{main}) that when the channel is good (i.e. high $h$) compared to the noise level ($\lambda$), the threshold converges to 
		 $\eta=-\log((1-\rho_{1})/(1-\rho_{0}))/\log(\rho_{1}/\rho_{0})$. On the other hand, in the presence of strong noise and weak channel condition, threshold becomes 1. We can observe the same in Fig \ref{hlow}. For instance, in Fig. \ref{hlow}(a), the decision boundary contains $m_{1}=49$ and $m_{2}=39$ for low value of $\lambda$, for which the simulated slope is 1.25 and theoretical slope $\eta=1.3$ (cf. (\ref{main})). Note that the slight mismatch is because $m_{i}$'s can take only integer values. Also in Fig. \ref{hlow}(b), for high values of $\lambda$, i.e., $\lambda=30$, $\eta$ approximately equals to one which corroborates with the simulated slope of 1.02, with $m_{1}=49$ and $m_{2}=48$. 
		  We can note that the ratio of $m_{1}$ and $m_{2}$ at the decision boundary remains the same for all values of $m_{1}$ and $m_{2}$, indicating that the decoder is ratio threshold decoder, when channel is known.
		  \begin{figure}[t]
	\centering
	\includegraphics[width=0.9\linewidth]{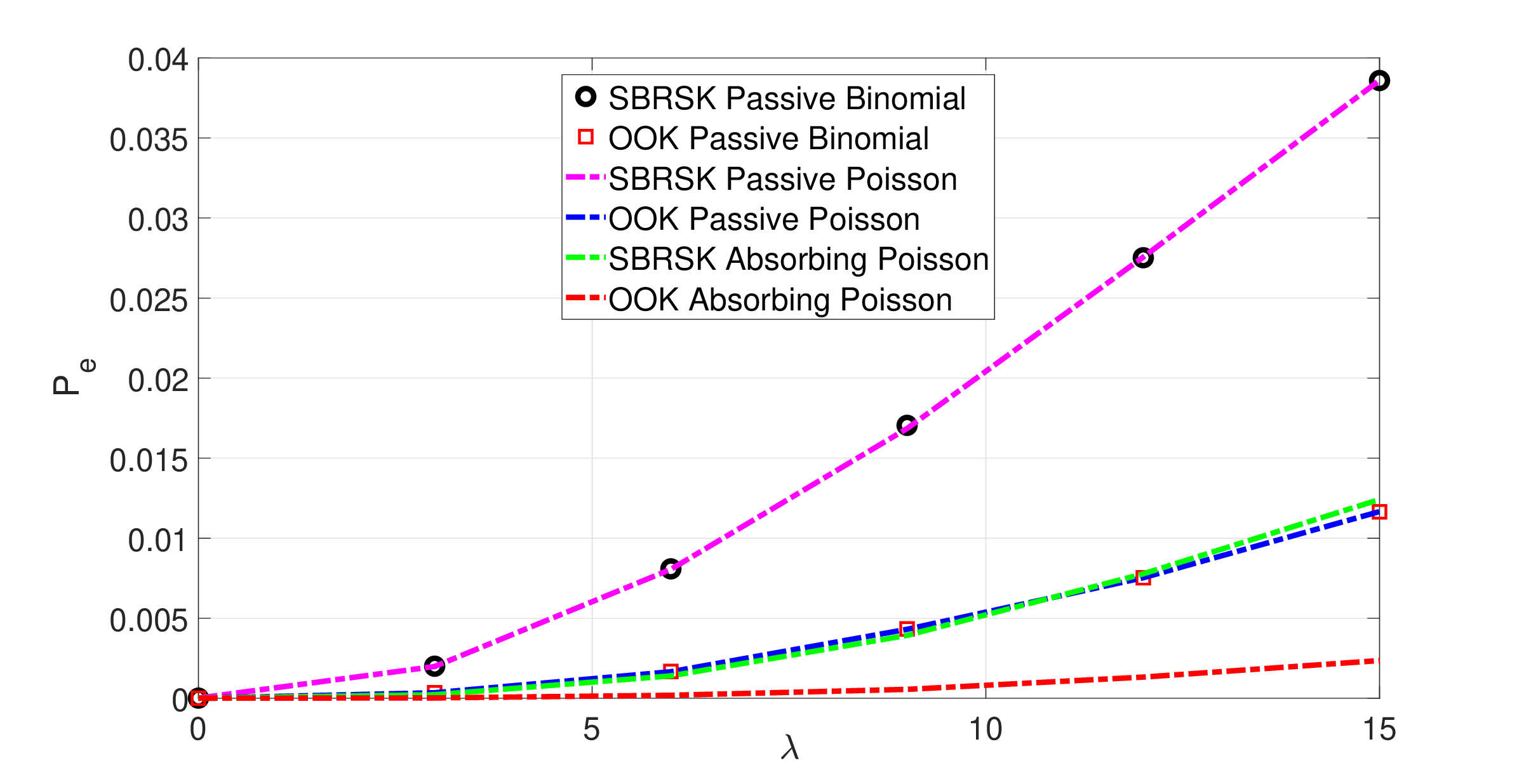}
	\caption{The variation of the probability of error with $\lambda$ for the SBRSK and OOK, along with the one under Poisson approximation. Here, $c=400$ and $N_{a}=2c=800$. Simulated BER matches very closely with the exact bit error under binomial reception of molecules.}
	\vspace{-1em}
	\label{noisem}
\end{figure}
\begin{figure}[t]
	\centering
	\includegraphics[width=0.9\linewidth]{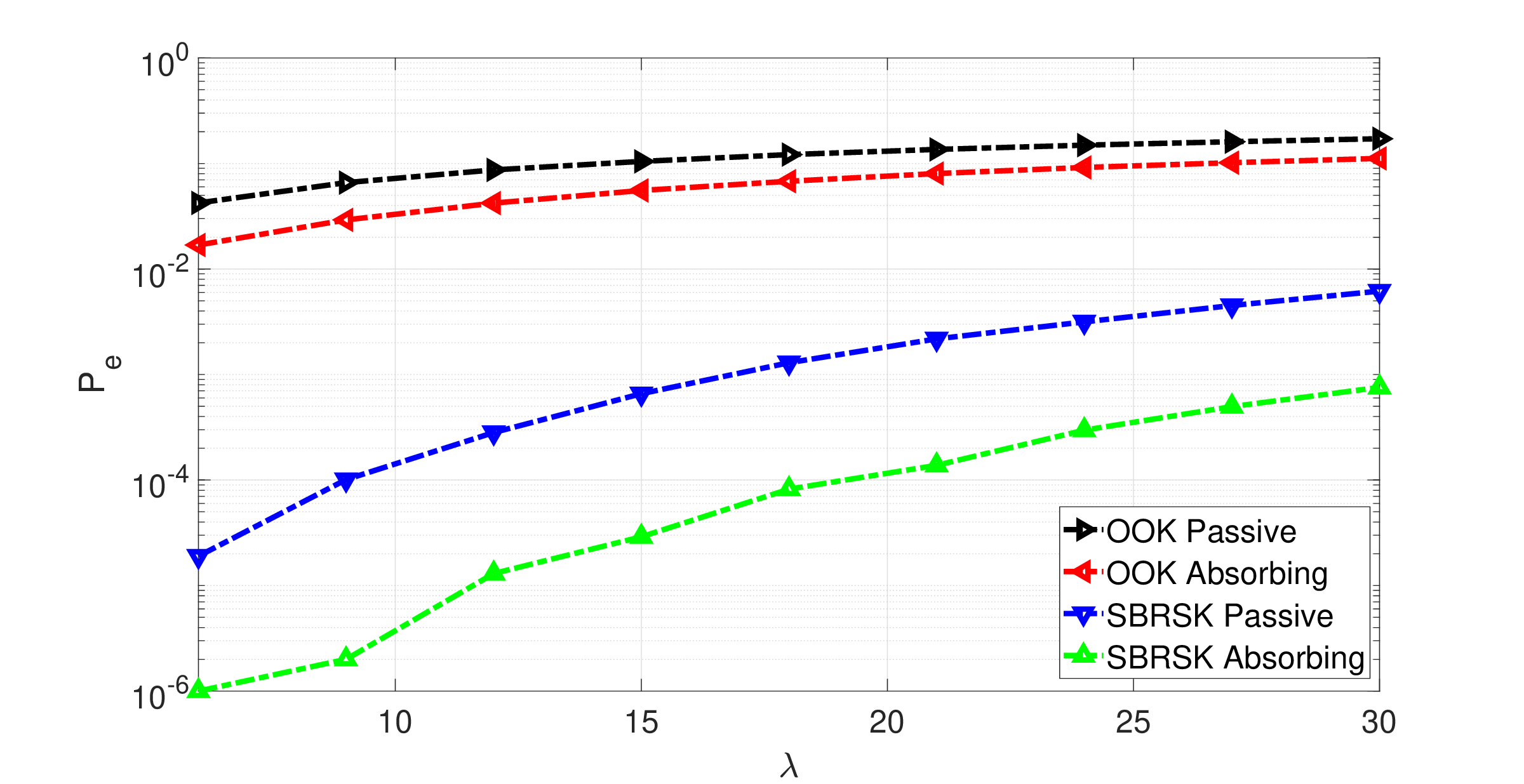}
	\caption{The variation of $P_e$ vs $\lambda$ for the SBRSK and OOK along with the one under Poisson assumption. Here, $c_{\text{RSK}}=2c=800$ and $c_{\text{OOK}}=c=400$.}
	\vspace{-1em}
	\label{morec}
\end{figure}
		\subsubsection{Performance of Symmetric BRSK}
		Next, we investigate the variation of bit error with $\lambda$ for passive and absorbing RX, under symmetric BRSK with $\rho_{0}=0$. Here, $c=400$ and $h_{\mathrm{P}}=0.0281$, which gives $ch_{\mathrm{P}}=11.24$. Also, for the absorbing RX case, we consider $T_{b}=10^{-5}$, which results in $h_{\mathrm{A}}=0.0369$ and consequently $ch_{\mathrm{A}}=14.76$. In Fig. $\ref{noisem}$, we plot the variation of bit error probability $P_{\mathrm{e}}$ versus $\lambda$ for SBRSK and OOK with binomial reception (using the decoder given in (\ref{cases})), along with $P_{\mathrm{e}}$ under Poisson approximation with decoder given in (\ref{cases2}). It is quite intuitive to observe that the bit error probability increases with $\lambda$. The Poisson approximation is accurate for the considered case since $h_{\mathrm{P}}$ and $h_{\mathrm{A}}$ is reasonably small. Moreover, it can be inferred that OOK performs better than SBRSK in static MC scenarios when the exact knowledge of the channel is available. This is because of the allocation of relatively lower number of molecules per type, i.e., $c$ in SBRSK as compared to $2c$ molecules for bit 1 in OOK. It can be observed in Fig. \ref{morec} that if we allocate more number of molecules to SBRSK, i.e., $2c$ in BRSK as compared to $c$ molecules for bit 1 in OOK, then SBRSK performs better than OOK for both passive and absorbing RX. However, we will show later that SBRSK still proves to be a better alternative in case of unknown CI even for lower molecular budget. 
			\begin{figure}[t]
			\centering
			\includegraphics[width=0.9\linewidth]{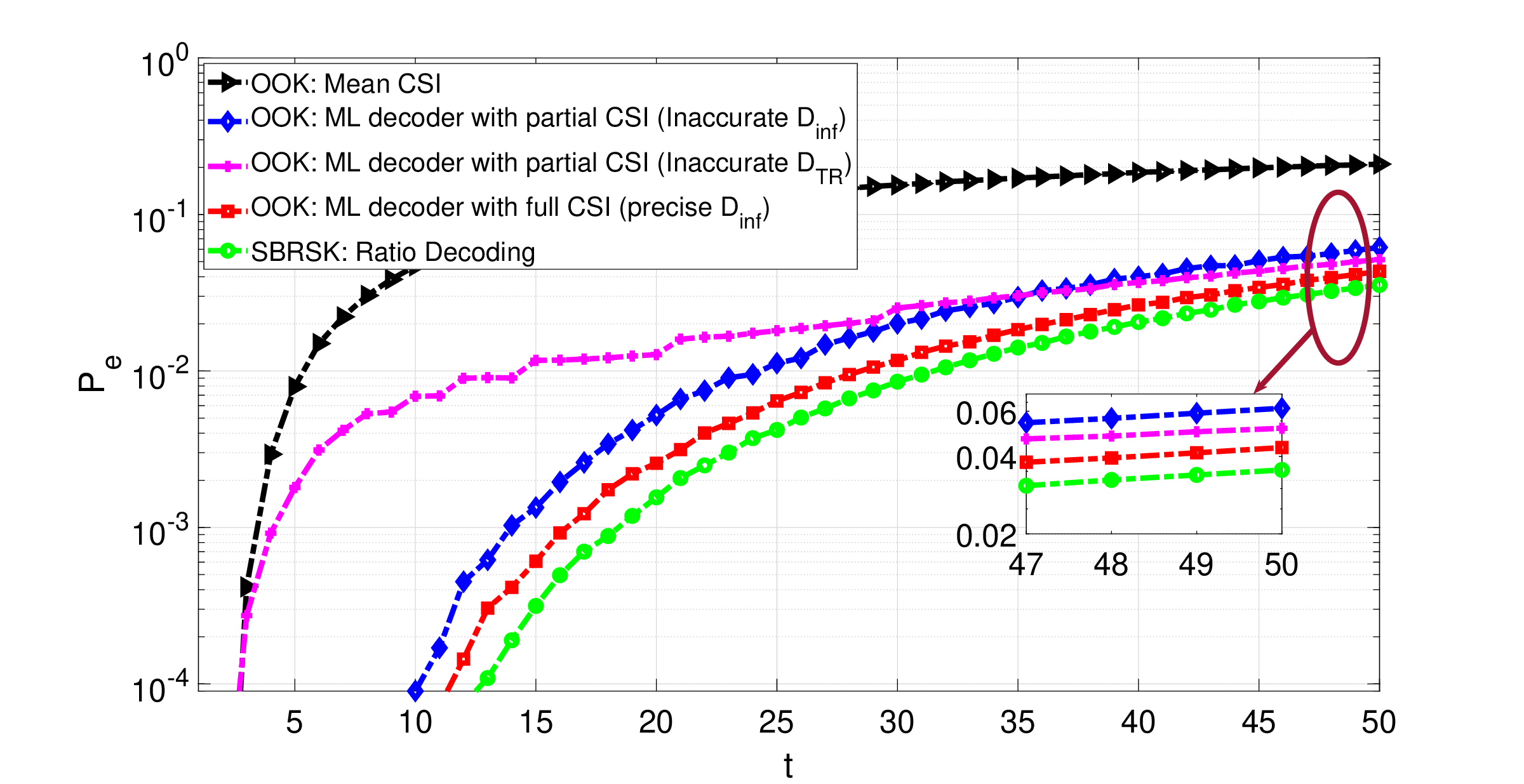}
			\caption{The variation of the $P_e$ with time for a dynamic MC system utilizing SBRSK, with decoder as given in (\ref{cases2}) and employing passive RX. Here, $D_{\text{inf}}=5 \times 10^{-9} \text{m}^{2}$/s, $\lambda=90$, and $D_\mathrm{TR}=5\times10^{-12} \text{m}^{2}$/s. The performance of OOK is also shown under mean CI and full CSI, utilizing decoders given in (\ref{meanhp}) and (\ref{dook3}), and under imperfect CSI with inaccurate $D_{\text{inf}}=9 \times 10^{-8} \text{m}^{2}$/s and inaccurate $D_\mathrm{TR}=5 \times 10^{-11} \text{m}^{2}$/s. The ML decoder with accurate/ inaccurate $D_{\text{inf}}$ utilizes full channel PDF for decoding the transmit symbol.}
			\label{meanf}
			\vspace{-1em}
		\end{figure}
		\vspace{-1em}
		\subsection{Dynamic MC Environment}
		\subsubsection{Performance of SBRSK} We now investigate the performance of mobile MC scenario for passive RX unit. Here, we take $c=2 \times 10^{5}$.
		The input stream of length $L=50$ is transmitted in one run and this process is repeated $10^{6}$ times. We have shown using (\ref{cases2}), that the optimal decoder for SBRSK is independent of the value of CIR, and hence remains the same for both mobile and static MC scenarios. Since channel varies with time for mobile MC system, Fig. \ref{meanf}, shows the variation of $P_e$ with time $t$ for both SBRSK and OOK, where time $t$ varies from $0$ to $L\tau_{s}$. Please note that we have compared the performance of OOK utilizing decoder under full CSI (as given in (\ref{dook3})) and mean CI (as given in (\ref{dook})). We can observe that the proposed SBRSK outperforms OOK due to the absence of the requirement of CI in determining decoding threshold.
		\subsubsection{Impact of inaccurate CSI}
		Fig. \ref{meanf} also shows the variation of the probability of bit error for the case when the receiver has an incorrect information of $D_{\text{inf}}$ or $D_\mathrm{TR}$ and performs decoding based on this imperfect CSI. As SBRSK does not require any CI, its performance is not affected due to imperfect CSI. 
		However, for OOK, the error increases significantly in the presence of imperfect CSI, arising from inaccuracies in either $D_{\text{inf}}$ or $D_\mathrm{TR}$. This is because OOK decoder is obtained by optimizing over the channel distribution, making it highly sensitive to the accuracy of the channel parameters, particularly $D_{\text{inf}}$ and $D_\mathrm{TR}$ \cite{stochastic}. In contrast, SBRSK decoding is independent of the channel. Fig. \ref{meanf} shows that symmetric BRSK outperforms OOK (with ML decoder), OOK with inaccurate $D_{\text{inf}}$, and approximate OOK (using mean threshold decoder), achieving relative performance gain of 18.37\%, 42.5\%, and 83.11\% respectively. 
		 It shows that OOK's error performance is quite sensitive to the threshold. 
		\begin{figure}[t]
			\centering
			\includegraphics[width=0.9\linewidth]{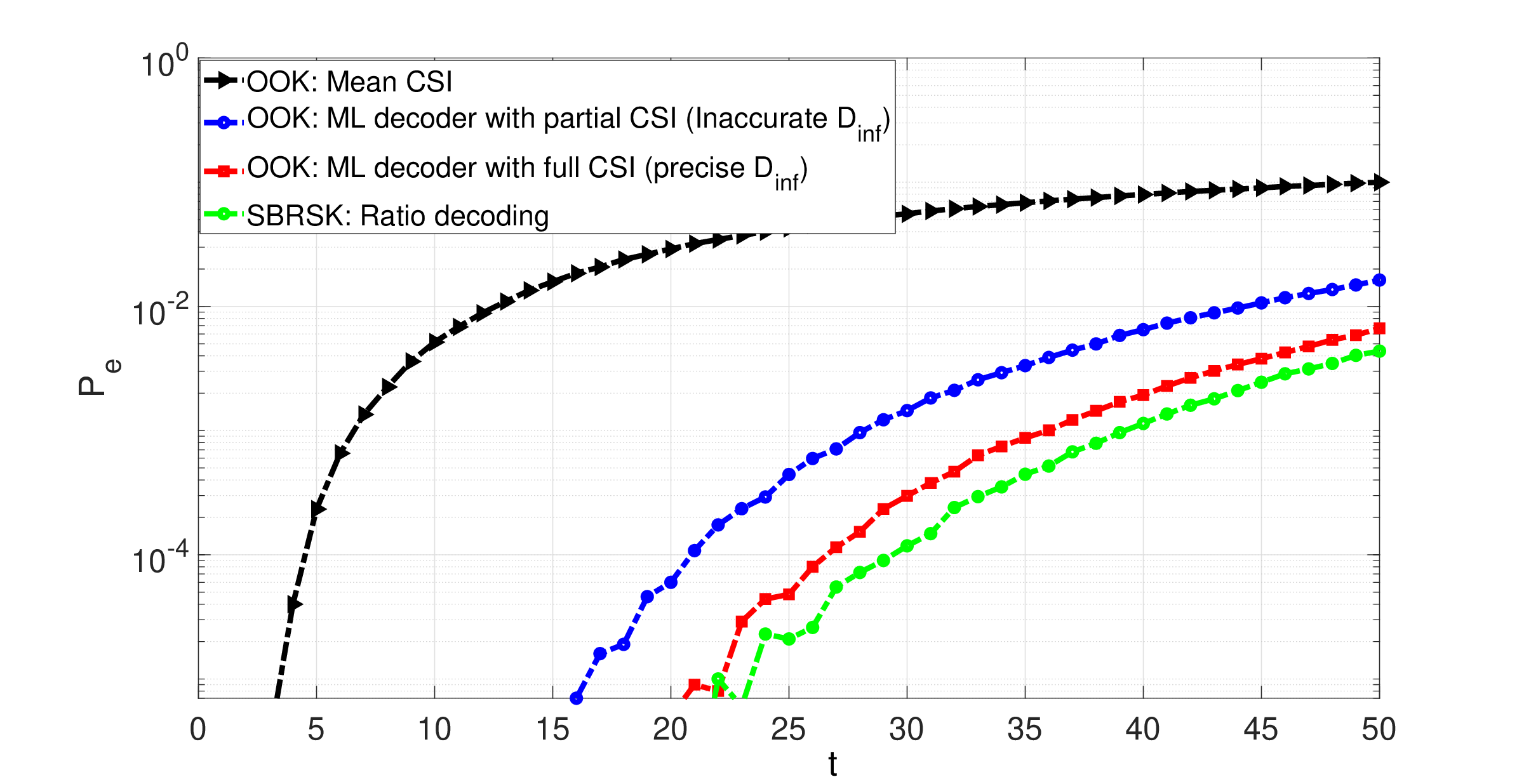}
			\vspace{-1em}
			\caption{The variation of the probability of error with time for a dynamic MC system utilizing SBRSK is provided, where decoder follows (\ref{cases2}) and an absorbing RX is considered with the following system parameters: $c=2 \times 10^{6}$, $\lambda=80$, $D_\mathrm{TR}=5\times10^{-11} \text{m}^{2}$/s, and $T_{b}=0.05$s.}
			\label{ab}
			\vspace{-1em}
		\end{figure}
		\begin{figure}[t]
			\centering
			\includegraphics[width=0.9\linewidth]{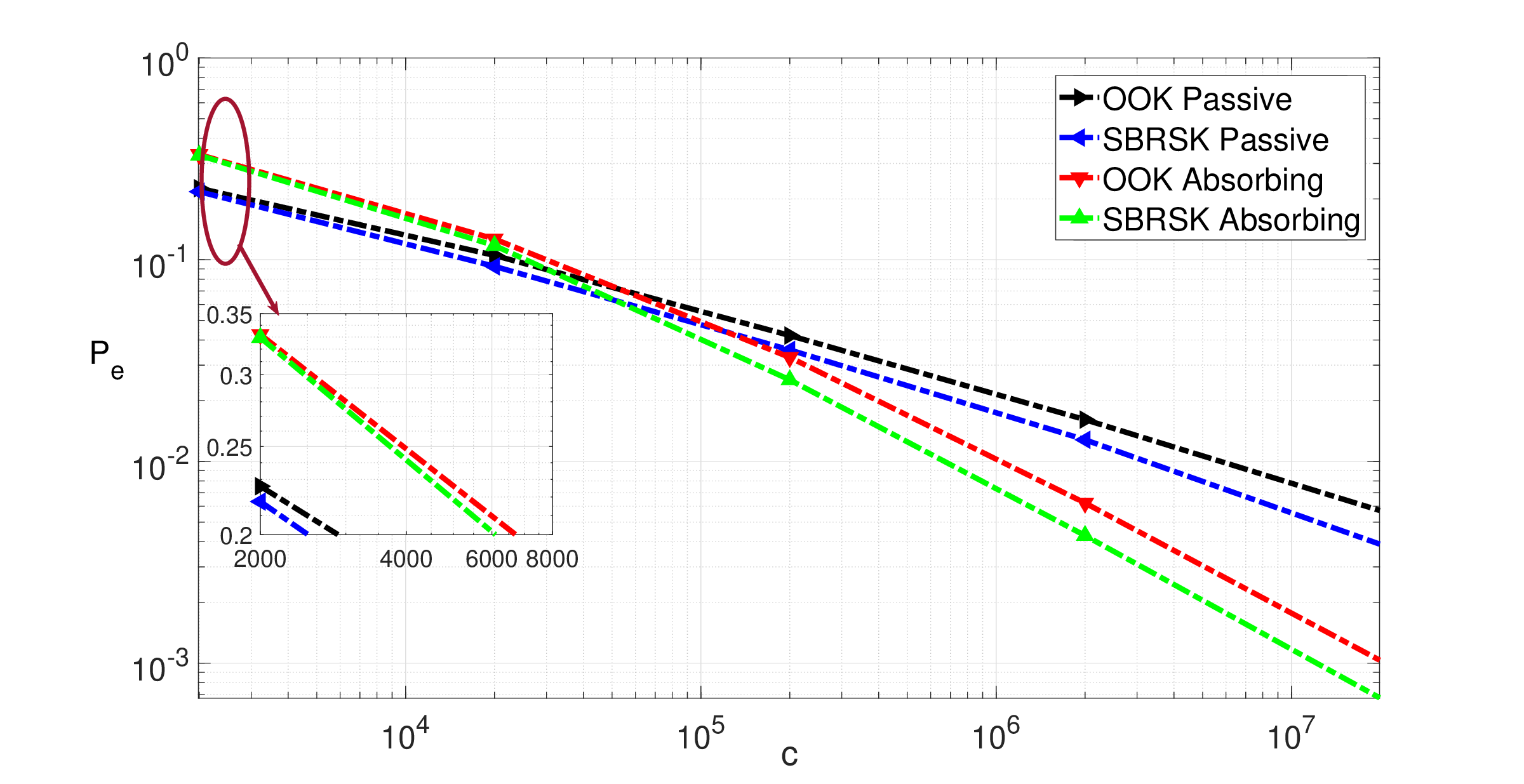}
			\vspace{-1em}
			\caption{The variation of $P_e$ with $c$ for a dynamic MC system utilizing SBRSK is provided, where the decoder follows (\ref{cases2}), and both passive and absorbing RX are considered.}
			\label{perrvsc}
			\vspace{-1em}
		\end{figure}
		
		Similar trends can be observed in Fig. \ref{ab} which shows the error probability of a mobile MC system with an absorbing RX. Here, the symmetric BRSK outperforms optimal OOK (with ML decoder), OOK with inaccurate $D_{\text{inf}}$, and approximate OOK (using mean threshold decoder), achieving the relative performance improvement of 34.6\%, 73.2\%, and 95.6\% respectively. It is interesting to note that we consider significant variation in the channel parameter, i.e., $D_\mathrm{TR}=5\times10^{-11}$, compared to $D_\mathrm{TR}=5\times10^{-12}$ in mobile passive RX case.
		\subsubsection{Impact of signal strength $c$~}
		Moreover, for poorly known or unknown channel conditions, we might not be able to improve much BER of OOK even for a very large value of $c$, while bit error performance of SBRSK can be substantially improved by increasing the number of transmitted molecules. This can be verified from Fig. \ref{perrvsc} that the gap between the error performance of SBRSK and OOK substantially improves with increasing values of $c$. Therefore, we can conclude from the figure that SBRSK outperforms optimally-decoded OOK in mobile MC system because of its robustness against time-varying channel fluctuations. 
		\begin{figure}[t]
			\centering
			\includegraphics[width=0.85\linewidth]{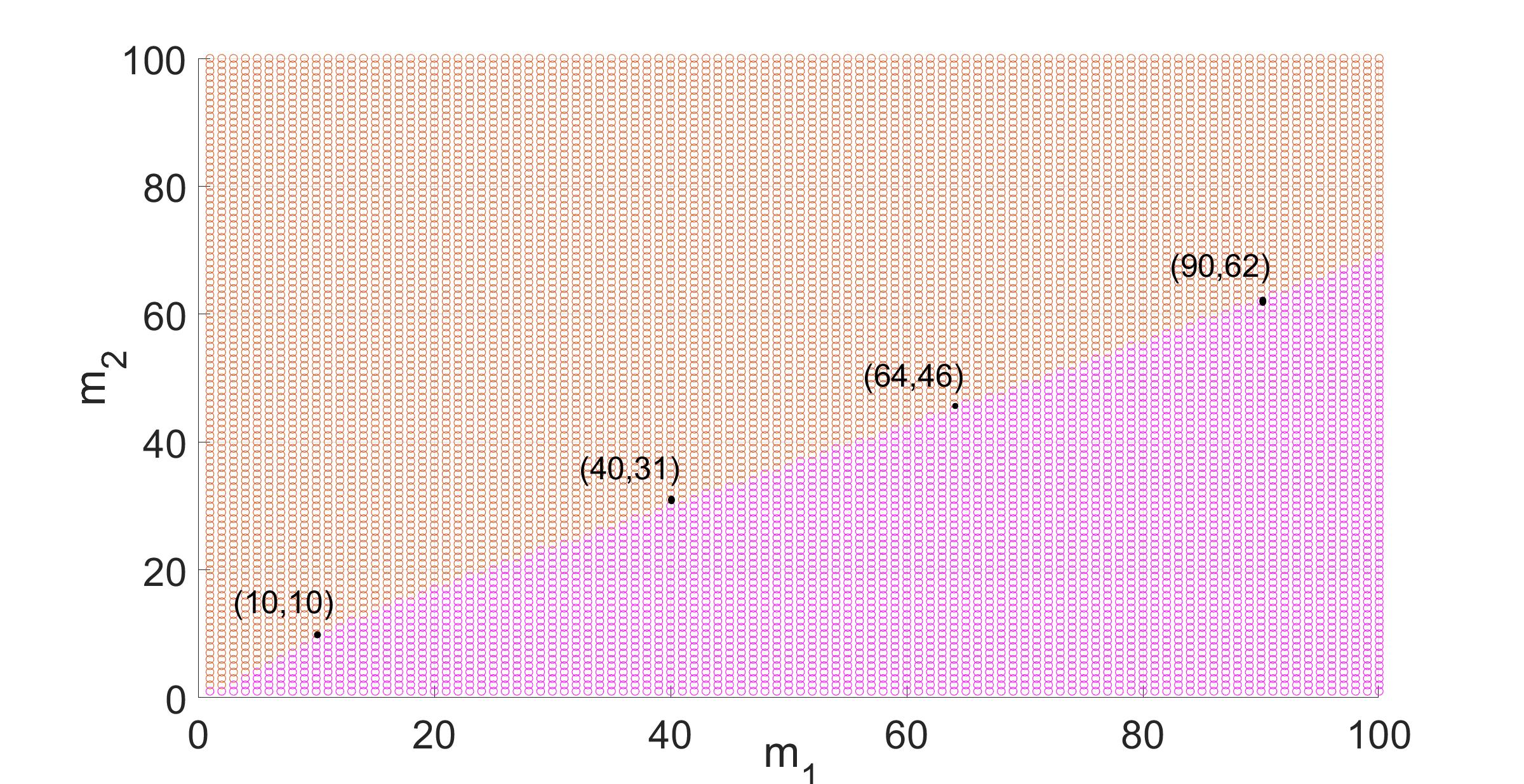}
			\vspace{-1em}
			\caption{Decoding region for bit 0 and bit 1 in a mobile passive RX utilizing BRSK with $\rho_{0}=0.4$ and $\rho_{1}=0.8$. The pink and orange region correspond to bit 1 and 0. Here, $D_\mathrm{TR}=5\times10^{-12} \text{m}^{2}$/s, $T_{b}=5\times10^{-4}$s and $t=50$s.}
			\label{asp}
			\vspace{-1em}
		\end{figure}
		\subsubsection{Decoding regions for BRSK} While decoding regions remain invariant to channel variation for SBRSK resulting in decoding regions separated by a line as decision boundary, the same is not true for assymmetric BRSK. Fig.  \ref{asp} shows the decoding regions for a dynamic MC system with a passive RX employing BRSK constellation, where the parameters are set to
		 $c=1000$, $\rho_{0}=0.4$ and $\rho_{1}=0.8$.
		 As expected, the decision boundary no longer is a straight line. It is interesting to note that large values of $m_{1}$ and $m_{2}$ correspond to the better channel conditions, where the effect of noise becomes negligible, the theoretical ratio threshold becomes $\eta=-\log((1-\rho_{1})/(1-\rho_{0}))/\log(\rho_{1}/\rho_{0})$, which evaluates to 1.58 for the considered parameters. However for small values of $m_{1}$ and $m_{2}$, where noise dominates, the threshold converges to 1. This behavior is evident in Fig. \ref{asp}, where the slope for small values of $m_{i}$'s is equal to $10/10=1$, whereas for large values it approaches $90/62=1.45$.
		In line with the analysis presented in Fig. $\ref{hlow}$, the theoretical and simulated values of slope are well matched. Although the decoding region slopes are comparable, the exact values of joint PMF of the received number of molecules for the mobile RX differs from that of the static RX. 
		\vspace{-1em}
		\section{Conclusion and Future Work}
		 This paper introduced a unified MAxCM framework that generalizes molecular modulation to higher-dimensions by jointly utilizing multiple molecular types. In this generalized framework, information can be encoded either in absolute molecular concentrations (concentration-based MAxCM) or in relative concentration ratios (ratio-based MAxRSK), providing a systematic lens through which existing schemes such as CSK, MoSK, and RSK emerge as special cases. 
		 
		 We then investigated two sub-classes of MAxCM. First, we presented rectangular MAxCM analogous to QAM and showed how the use of multiple axes can provide better BER and SE. Next, we discussed an important subclass of MAxCM, i.e., MAxRSK, and showed that if appropriately designed, it can eliminate the need for explicit CI at the RX. We derived the corresponding ML decoder and showed that it reduces to an optimal weighted combining rule. Furthermore, we analyzed a special case of MAxRSK, in particular, SBRSK, and showed that its constellation design admits a channel-independent decoding even in the presence of noise.
		 We also showed the robustness of SBRSK in dynamic
		 channel conditions compared to OOK while maintaining
		 a comparable decoding complexity under constraints on
		 molecular budget and bit duration, making it highly attractive not only for dynamic
		 MC scenarios but also for settings where CI is difficult or
		 infeasible to obtain.
		 Extensive numerical results demonstrated that SBRSK consistently outperforms conventional OOK in dynamic MC scenarios. This robustness persists even when only partial or imperfect CSI is available, in contrast to OOK, whose performance degrades sharply due to its sensitivity to CI.
		   Even in static MC channels with unknown parameters, SBRSK remains advantageous owing to its simplicity and robustness. Furthermore, analysis of asymmetric BRSK demonstrated that the simulated decoding thresholds corroborate well with the analytical values, validating the tractability of the proposed framework. The impact of the RX type was also examined, revealing that while both passive and absorbing receivers benefit from SBRSK, absorbing RXs amplify the relative performance gains of SBRSK, especially in mobile environments. These results highlight the importance of jointly considering modulation design and RX type in MC system optimization. We demonstrated that SBRSK naturally extends to higher-order (one example is MAxRSK(3,3)), while preserving
		   channel-independent decoding. 
		 
		 Collectively, the findings of this work prove that the proposed SMAxRSK is a low-complexity and robust approach, particularly well suited for dynamic and biologically realistic MC scenarios where accurate channel estimation is difficult or infeasible. At the same time, concentration-based MAxCM constellations offer flexibility for higher-order design.

		 Several promising research directions naturally arise from this work. One important avenue is the systematic optimization of higher-order MAxRSK and MAxCM constellations, including the design of symbol sets that balance robustness, spectral efficiency, and molecular resource consumption. Another direction is extending the analysis to heterogeneous mobility and propagation conditions, and
		 non-ideal environments including intra-body, soil, or microfluidic channels, where diffusion statistics may deviate from the standard models. On the application side, ratio-based modulation schemes hold significant potential for bio-cyber interfaces and synthetic biological systems, where simplicity, robustness, and minimal reliance on channel knowledge are paramount. 
		\vspace{-1.1em}
		\bibliography{IEEEabrv,RSK_reference}
		\bibliographystyle{IEEEtran} 
	\end{document}